\documentclass[twocolumn,citeautoscript,superscriptaddress,prx]{revtex4-2}

\usepackage{amsthm}
\usepackage{graphicx}
\usepackage{array}
\usepackage{color}
\usepackage{upgreek}
\usepackage{float}
\usepackage{enumerate}
\usepackage{enumitem}
\usepackage{lipsum}
\usepackage{booktabs}
\usepackage{url}
\usepackage{braket}
\usepackage{bm}
\usepackage{bbm}
\usepackage[T1]{fontenc}
\usepackage{scrextend}
\usepackage{hyperref}
\usepackage[dvipsnames]{xcolor}
\definecolor{C1}{RGB}{52, 89, 149}
\definecolor{C2}{RGB}{251, 77, 61}
\definecolor{C3}{RGB}{3, 206, 164}
\definecolor{C4}{RGB}{202, 21, 81}
\hypersetup{colorlinks=true, linkcolor=C2, citecolor=C2, urlcolor=C2}
\usepackage{dsfont}
\usepackage{epstopdf}
\usepackage[export]{adjustbox}
\usepackage{thmtools}
\usepackage{amsmath}
\usepackage{thm-restate}
\usepackage[caption=false, list=true]{subfig}  

\usepackage{fdsymbol}

\usepackage{accents}

\theoremstyle{remark}

\newcommand*{\ups}{\Upsilon}
\newcommand*{\tree}{\Upsilon^{\mathrm{(tree)}}}
\newcommand*{\s}{\mathcal{H}_S}

\newcommand*{\ot}{\otimes}
\newcommand*{\nn}{\nonumber}

\newcommand*{\id}{\mathds{1}}

\newcommand*{\llangle}{\langle \! \langle}
\newcommand*{\rrangle}{\rangle \! \rangle}

\newcommand*{\mc}{\mathcal}
\newcommand*{\dg}{\dagger}
\newcommand*{\ex}{\mathrm{e}}

\DeclareMathOperator{\tr}{tr}

\usepackage{pict2e}
\makeatletter
\newcommand{\bigcomp}{%
  \DOTSB
  \mathop{\vphantom{\sum}\mathpalette\bigcomp@\relax}%
  \slimits@
}
\newcommand{\bigcomp@}[2]{%
  \begingroup\m@th
  \sbox\z@{$#1\sum$}%
  \setlength{\unitlength}{0.9\dimexpr\ht\z@+\dp\z@}%
  \vcenter{\hbox{%
    \begin{picture}(1,1)
    \bigcomp@linethickness{#1}
    \put(0.5,0.5){\circle{1}}
    \end{picture}%
  }}%
  \endgroup
}
\newcommand{\bigcomp@linethickness}[1]{%
  \linethickness{%
      \ifx#1\displaystyle 2\fontdimen8\textfont\else
      \ifx#1\textstyle 1.65\fontdimen8\textfont\else
      \ifx#1\scriptstyle 1.65\fontdimen8\scriptfont\else
      1.65\fontdimen8\scriptscriptfont\fi\fi\fi 3
  }%
}
\makeatother

\begin{document}


\title{{Capturing long-range memory structures with tree-geometry process tensors}}

\author{Neil Dowling}
\email[]{neil.dowling@monash.edu}
\address{School of Physics \& Astronomy, Monash University, Clayton, VIC 3800, Australia}

\author{Kavan Modi}
\email[]{kavan.modi@monash.edu}
\address{School of Physics \& Astronomy, Monash University, Clayton, VIC 3800, Australia}
\affiliation{Quantum for New South Wales, Sydney, NSW 2000, Australia}

\author{Roberto N. Mu\~noz}
\address{School of Physics \& Astronomy, Monash University, Clayton, VIC 3800, Australia}

\author{Sukhbinder Singh}
\address{School of Physics \& Astronomy, Monash University, Clayton, VIC 3800, Australia}
\address{Multiverse Computing, Spadina Ave., Toronto, ON M5T 2C2, Canada}

\author{Gregory A. L. White}
\email{gregory.white@fu-berlin.de}
\address{School of Physics \& Astronomy, Monash University, Clayton, VIC 3800, Australia}
\affiliation{Dahlem Center for Complex Quantum Systems, Freie Universit\"at Berlin, 14195 Berlin, Germany}

\begin{abstract}
We introduce a class of quantum non-Markovian processes -- dubbed {process trees} -- that exhibit polynomially decaying temporal correlations and memory distributed across time scales. This class of processes is described by a tensor network with tree-like geometry whose component tensors are (1) {causality-preserving} maps (superprocesses) and (2) {locality-preserving} temporal change of scale transformations. We show that the long-range correlations in this class of processes tends to originate almost entirely from memory effects, and can accommodate genuinely quantum power-law correlations in time. Importantly, this class allows efficient computation of multi-time correlation functions. To showcase the potential utility of this model-agnostic class for numerical simulation of physical models, we show how it can efficiently approximate the strong memory dynamics of the paradigmatic spin-boson model, in terms of arbitrary multitime features. 
In contrast to an equivalently costly matrix product operator (MPO) representation, the ansatz produces a fiducial characterization of the relevant physics. Finally, leveraging 2D tensor network renormalization group methods, we detail an algorithm for deriving a process tree from an underlying Hamiltonian, via the Feynmann-Vernon influence functional. Our work lays the foundation for the development of more efficient numerical techniques in the field of strongly interacting open quantum systems, as well as the theoretical development of a temporal renormalization group scheme. 
\end{abstract}

\maketitle

\section{Introduction}

Correlations and complexity are intimately intertwined. It is generally the case that for simple, ordered systems, correlations are easy to model. On the other hand, although chaotic systems can become highly correlated across large space and time scales, the correlations between local degrees of freedom can often be simply described using the tools of statistical mechanics~\cite{Deutsch1991, Srednicki,kadanoff2000statistical}. Somewhere between these two extremes lies the interesting case which can be difficult to model: \textit{complex} dynamics~\cite{Crutchfield2012,gu_quantum_2012,binder_practical_2018,ghafari_dimensional_2019}. A characteristic feature of complex systems is long-range correlations -- a power law spectrum or $1/f$ noise.\textsuperscript{\footnote{Not all systems with long-rage correlations are complex, e.g. the GHZ state, or possibly unitary (Markovian) quantum processes.}} 
{This complexity property is ubiquitous across fields of science, from criticality in condensed matter physics~\cite{stanley_introduction_1987,sachdev} and random networks~\cite{doi:10.1126/science.286.5439.509,stauffer_introduction_2014} to characteristics of the human brain~\cite{Haimovici2013} and DNA sequences~\cite{li_study_1997,PENG1995180}.}

A central challenge in modern physics is isolating the key properties of correlated quantum systems; taming complex systems into their essential physics. For instance, numerical techniques based on tensor networks -- such as tree tensor networks (TTNs)~\cite{Shi2006, Tagliacozzo2009,Silvi2010,Murg2010} and the multi-scale entanglement renormalization ansatz (MERA)~\cite{vidal_class_2008,evenbly_algorithms_2009} have constituted groundbreaking progress in the development of quantum many-body physics, accurately modeling critical states, such as ground states of gapless Hamiltonians. 
As it stands, no such equivalent exists in the dynamical setting. However, quantum combs~\cite{chiribella_theoretical_2009} and process tensors~\cite{processtensor,milz2020quantum} provide a natural framework to study multi-time processes by mapping them to many-body states. But even equipped with these space-time dualities, the translation of aforementioned tensor network results to the temporal (or spatiotemporal~\cite{Dowling2022}) regime is highly non-trivial.
In open dynamical systems, temporal correlations can be mediated by a strongly-interacting -- but inaccessible -- bath. 
The resulting model is hence both mixed and subject to causal order requirements.
That is to say, correlations are carried forward in time by an external bath, not {just via} the system itself. 
As a result, there is a glaring gap in the description of open quantum systems with power-law temporal correlations, such as in spin-boson models~\cite{le2008entanglement, LeHurrBook2010}, Floquet time crystals~\cite{ivanov2020feedback}, and open solid-state systems subject to complex noise~\cite{Paladino2014}. Such cases constitute monumental simulation challenges, with (at worst) exponential growth in both the spatial and temporal degrees of freedom.

In this work, we are concerned with {addressing} this gap. Specifically, we build a class of quantum non-Markovian processes to efficiently represent processes that exhibit strong and slowly (polynomially)-decaying temporal correlations. We then apply our methods to a prototypical model, the spin-boson model, to showcase their high efficacy over usual multitime models.

A key lesson from tensor network theory is that the geometry of a tensor network limits the structure and scaling of correlations in the many-body state it represents~\cite{evenbly_tensor_2011}. For instance, a matrix product state (MPS), which has a linear geometry, generally exhibits exponentially decaying correlations~\cite{Eisert2010area,brandao2015exponential}. Meanwhile, TTNs and MERA -- both of which are hierarchical tensor networks extending to a hyperbolic geometry -- naturally accommodate polynomially decaying \emph{spatial} correlations. 
Our class of models takes inspiration from the TTN structure, and we hence refer to them as \emph{process trees}. Their general form is depicted in Fig.~\ref{fig:treeintro}. A key result in this paper is showing that process trees generically exhibit power-law decay for temporal correlations \textit{and} non-Markovianity. {That is, beyond a direct power-law correlation on some system of interest, the memory as mediated by an inaccessible bath has an influence on the system that also decays slowly}. Moreover, process trees can compute correlations between time-local operators with polynomial cost, making the analysis efficient. We substantiate these claims with a series of analytic results and numerical calculations. Although we take inspiration from the spatial TTN geometry, the analogy ends there as the internal structures of the process tree are quite different, stemming from temporal causality constraints.

While the process tree is designed to capture a specific \emph{structure} of temporal correlations, it is not a priori obvious that this implies it is relevant to real physics. In light of this, we showcase how these characteristics may be applied to the study of relevant physical systems by demonstrating that the ansatz can be used to faithfully represent the spin-boson model across a critical phase transition of Berezinskii–Kosterlitz–Thouless (BKT) type~\cite{Bulla_2003,le2008entanglement,LeHurrBook2010}.
This system models physically relevant impurity setups, and generally exhibits long-range temporal correlations. In fact, the bond dimension of the approximate influence matrix for the Ohmic spin-boson model can be shown to have polynomially growing bond dimension~\cite{vilkoviskiy2023bound}, characteristic of critical (power-law) temporal correlations. 
We take a variational ansatz for the process tree and fit it to the true spin-boson model with an optimization approach.
That the tree fits well is in contrast to taking a matrix product operator (MPO) fit -- with a greater number of free parameters -- which is much less capable of describing the multi-time physics across the phase transition. The upshot here is that tailoring the geometry of the model to expectations of the physics permits both a more efficient and a more suitable representation of complex processes. 
Moreover, not only does a process tree fit this model well, but it additionally generalizes. 
Once we determine the elementary building block of the process tree, we can use it to construct large-scale processes. We demonstrate this again for the spin-boson model with surprisingly high efficacy. This suggests that some time and scale invariant properties about the dynamics may be learned and later applied to understand greater instances of those systems. {Finally, we lay out a method for constructing a process tree from an underlying Hamiltonian, combining techniques from the celebrated influence matrix approach to multitime physics~\cite{Banuls2009,Hastings2015,Strathearn2018,Lerose2021,Lerose2021prb}, with the tensor renormalization group method of 2D tensor networks~\cite{Levin_2007,Vidal2007,PhysRevB.78.205116,Evenbly2015,Evenbly_2017}. Therefore, we show both that our ansatz is expressive of relevant systems, and that it can be systematically produced from an underlying model.  }

\begin{figure}[t]
    \centering
    \includegraphics[width=\columnwidth]{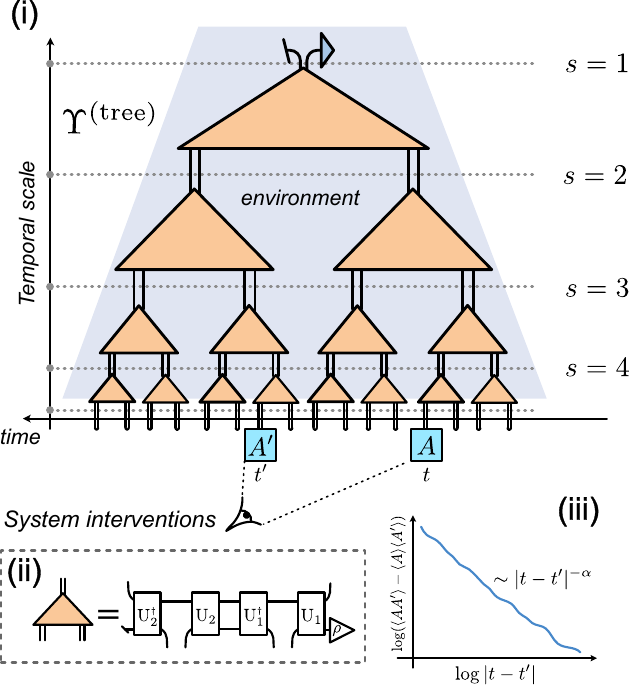}   \captionsetup{justification=justified,singlelinecheck=false} 
    \caption{
    Depiction of the essential ideas found in this work. 
    (i) A process tree describing an open quantum system $S$ interacting with its environment. The process is constructed by a series of iteratively-connected maps (tensors). These maps implement a temporally-consistent fine-graining operation where one intervention at a single time is mapped to two interventions at two times. The vertical extent of the tensor network hence corresponds to a timescale $s$ where interventions may be exponentially more frequent. 
    Each pair of open indices at the bottom of the tree, and pairs of open indices intersected by the dashed lines at any time scale, is an intervention slot where an instrument may be applied. 
    (ii) Each causality-preserving map is parameterized as shown by two unitary maps $U_2$ and $U_1$ and a (vectorized) density matrix $\rho$. 
    (iii) Indicative figure of the hallmarks of polynomially decaying correlations captured by the tensor network. 
    }
    \label{fig:treeintro}
\end{figure}

Consequently, the process tree also serves as a practical path to better theoretical and numerical analyses of complex quantum processes. These results add to the growing number of studies analyzing the richness of open quantum systems, beyond the weak coupling and Markovian regime~\cite{Rivas2012,Wiseman2018,milz2020quantum,PhysRevLett.107.160601,PhysRevLett.118.100401,Lambert,Inchworm,Inchworm2,Ebler}. Indeed, the introduction of process tensors~\cite{processtensor, processtensor2, milz2020quantum} permitted an operational description of multi-time statistics whereby temporal correlations are mapped onto spatial ones, spurring the nascent study of many-time physics~\cite{White2021,Milz2021GME}. Recently, MPO representations have been the chief tool of choice in taming inefficient representations, from cutting-edge numerical techniques~\cite{Banuls2009,Hastings2015,Lerose2021,Lerose2021prb,Strathearn2018,Gu2018,Gu2018A,Pollock2019,pollockscipost,white_non-markovian_2022,gribben_exact_2022, fowler-wright_efficient_2022,Cygorek2022,gribben_using_2022,link2023open,Bertini2023,Thoenniss2023,fux_tensor_2023,ng_real-time_2023} to experimental reconstruction of multi-time processes in a laboratory setting~\cite{Guo2020,White2021}. However, such methods only lead to relative efficiency when correlations grow slowly with time, unless extra approximations and fine-tunings are made.

Our work instantiates a conceptual re-evaluation of the structure of temporal correlations and presents a framework by which scale and renormalization may be understood in a dynamical context. In the spatial case, the fact that TTNs and MERA capture properties of the state at different length scales can be used to extract universal properties of the state, e.g., in the case of critical ground states, the underlying conformal field theory data~\cite{CriticalMERA}. Moreover, these methods can be used to pinpoint quantum phase transitions by computing the fixed points of the renormalization map. The process tree analogously opens up the possibility of developing systematic tools for temporal coarse-graining to identify different phases of quantum processes and their universal properties, i.e., the universality in long-time quantum dynamics.
Indeed, we find that a process tree, endowed with hyperbolic geometry, {is better suited for representing a spin-boson process 
than a linear network.}

The remainder of the paper is organized as follows. In Sec.~\ref{sec:background}, we review basic concepts pertaining to quantum processes and superprocesses alongside graphical notation which is used extensively in this work. In Sec.~\ref{sec:treeconstruct}, 
we construct from first principles our process tree descriptor. The key constituent here is a temporal fine-graining transformation alongside causal consistency conditions.
In Sec.~\ref{sec:multitime}, 
we turn to analyzing the scaling and computational cost of computing multi-time correlation functions from process trees. We prove that generic process trees always exhibit polynomially decaying two-time correlation functions, and how a `causal structure' emerges in the time scale direction of the network, which helps reduce the computational cost and complexity of their software implementation. 
In Sec.~\ref{sec:NM-cor} we consider specifically the non-Markovian properties of process trees, showing that they capture critical correlations mediated by both classical and genuinely quantum baths. {We then turn to showcasing the relevance and applicability of this ansatz to physical models: first in Sec.~\ref{sec:grug}, using variational fitting techniques we demonstrate process trees to be both a suitable and efficient descriptor of the many-time physics found in the spin-boson model across a critical phase transition. Finally, in Sec.~\ref{sec:plaquette} we show how to systematically construct a process tree from an underlying Hamiltonian.} 

\section{Background}\label{sec:background}
{We will now briefly} review the \emph{process tensor} framework~\cite{chiribella_theoretical_2009,processtensor,processtensor2,milz2020quantum}. This is the operational basis for describing multi-time temporal correlation functions in any dynamical open quantum system, including arbitrary non-Markovian phenomena. In doing so, we encounter a hierarchy of increasingly complex objects: quantum states $\rightarrow$ operators $\rightarrow$ channels $\rightarrow$ processes (multi-time channels) $\rightarrow$ superprocesses (maps of multi-time channels). All of these objects are basically instances of \textit{tensors}, namely, multi-dimensional arrays of numbers. Throughout this section, we also introduce the graphical representation of tensor networks, which is used extensively in this paper to provide a compact yet precise representation of expressions with potentially many indices. A more complete introduction to graphical notation and tensor networks can be found in App.~\ref{ap:tn}, and in selection of comprehensive reviews Refs.~\cite{wood_tensor_2015,Bridgeman2017,Orus2019, Cirac2021}. Further details on the process tensor framework can be found in the tutorial of Ref.~\cite{milz2020quantum}.

\subsection{Non-Markovian Quantum Processes} \label{sec:PT}
Consider a controlled quantum system \textit{S} that is interacting with an inaccessible and uncontrollable environment \textit{E}, taken to be described together by the finite-dimensional Hilbert spaces ${\mc{H}_S} \otimes {\mc{H}_E}$. This is the standard setup for open quantum systems.

Without loss of generality,\textsuperscript{\footnote{One can always add a finite-dimensional ancilla to the $SE$ system so that the total isolated system evolves unitarily (Stinespring dilation).}} we assume that the system and environment together constitute a closed system that evolves unitarily in time under the action of a unitary map $\mathcal{U}$,
\begin{equation}\label{eq:schrodinger}
\rho_t = \mathcal{U}_t(\rho_0) = u^\dagger_t (\rho_0) u_t = \mathrm{U}_t |\rho_0\rangle \!\rangle, 
\end{equation}
where $u_t$ is a unitary matrix parametrized by time $t$, and the final equality corresponds to the Liouville superoperator representation, with $\mathrm{U}:=u^* \otimes u$ acting through left matrix multiplication on the vectorized density matrix $|\rho_0  \rrangle$~\cite{OperationalQDynamics}. {Within this (Liouville) representation, when referring to a `system Hilbert space' we mean the doubled space $\mc{H}_S \otimes \mc{H}_S\cong \mathcal{B}(\mathcal{H}_S)$, wherein elements are vectorized density matrices.}

During its evolution, one may intervene instantaneously on the system \textit{S} by applying an arbitrary time-ordered sequence of instruments, leading to a general expression for a multitime correlation function,
\begin{align}
    \braket{\mathcal{A}_{x_k} \cdots \mathcal{A}_{x_1}}_{\ups} &= \tr[\mathcal{A}_{x_k} \mc{U}_{k-1} \cdots \mathcal{A}_{x_2} \mc{U}_1 \mathcal{A}_{x_1} (\rho)] \nn \\
    &=  \langle \! \langle \id | \mathrm{A}_{x_k} \mathrm{U}_{k-1} \dots \mathrm{A}_{x_2} \mathrm{U}_1 \mathrm{A}_{x_1} | \rho \rangle \!\rangle \label{eq:PT}\\
    &= \includegraphics[scale=0.85, valign=c]{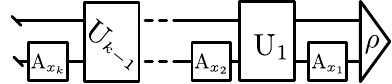}. \nn
    \end{align}
Here, $\mc{U}_i$ represents the global $SE$ dynamics {for time $(t_{i+1}-t_{i})$,} and the top (bottom) wire represents the time-evolving space $\mc{H}_E$ ($\mc{H}_S$). The sequence of system observables are represented by instruments $\{\mathcal{A}_{x_j}\}_{j=1}^k$ with outcomes $\{x_j\}$, where subscript $j$ denotes the time $t_j$.\textsuperscript{\footnote{Here, we have assumed that the interventions act independently of one another, that is, $\mathbf{A}_{k} = \mathrm{A}_{k} \otimes \mathrm{A}_{k-1} \otimes \dots \otimes \mathrm{A}_{1}$, but they can generally be correlated. Graphically, correlated instruments correspond to introducing additional wires between the $A$-boxes, and are called testers~\cite{Chiribella_2008}.}}

An instrument $\mathcal{A}_{x_j}$ is a completely positive trace non-increasing map, with $ \mathrm{A}_{x_k} $ in the second line of Eq.~\eqref{eq:PT} its Liouville superoperator representation. The same distinction holds between $\mathcal{U}$ and $\mathrm{U}$. In this representation the composition of maps simply corresponds to matrix multiplication and therefore can be represented graphically through tensor contractions, as in the final line of Eq.~\eqref{eq:PT}. We will work almost exclusively in this representation in this work. We use the convention of time running from right to left, such that state vectors/kets (dual vectors/bras) have open wires to the left (right),
\begin{equation}
    | \rho \rrangle = \includegraphics[scale=0.85, valign=c]{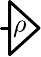}, \quad \llangle \id | = \includegraphics[scale=0.85, valign=c]{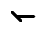},
\end{equation}
where we have also introduced the special notation of an oblique line for the identity projection, corresponding to a (partial) trace. Notice that wires in this representation are {squared} in size, e.g. dimension $4$ for a qubit system $\mc{H}_S$, encoding all degrees of freedom of a {(possibly)} mixed state. See App.~\ref{ap:tn} for further details on graphical notation. 

In computing Eq.~\eqref{eq:PT}, it is not always possible to trace out the environment to obtain a single CPTP map, or even a time-ordered sequence of \textit{uncorrelated} CPTP maps, that completely describe the time-evolution of the system without referencing the environment. Processes that admit such descriptions are called \textit{Markovian} (memoryless). In general, however, the environment retains a memory of the system's past, which influences (correlates with) the future evolution of the system. In this non-Markovian case, it is desirable to have an exact description of all correlations Eq.~\eqref{eq:PT} for any chosen set $\{\mathcal{A}_{x_j}\}_{j=1}^k$.

Through the lens of tensor networks, it is simple to separate the outside interventions $\{ \mc{A}_{x_j} \}$ from the uncontrollable parts of the full system-environment dynamics $\{ \mc{U}_j, \rho \}$ in expression \eqref{eq:PT} to get $\braket{\mathcal{A}_{x_k} \cdots \mathcal{A}_{x_1}}_{\ups} = \langle \! \langle{\ups | \mathbf{A}_{k} }\rangle \! \rangle= \tr[\ups \mathbf{A}_{k}^T  ]${, where the $T$ superscript indicates a transpose}. Here, we implicitly defined the process tensor $\ups_{k}$ and the multitime instrument $\mathbf{A}_k$,
    \begin{align}
            \mathbf{A}_{k} =& \mathrm{A}_{x_k} \otimes \dots \otimes \mathrm{A}_{x_1} \\ 
            =&\quad \includegraphics[scale=0.85, valign=c]{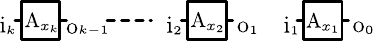}, \nn \\
            \text{ and}
            \nn \\
            \ups_{k} =& \tr_E [\mathrm{U}_{k-1} \star \dots \star \mathrm{U}_1 \star \rho] \label{eq:pt_link} \\
            =&\includegraphics[scale=0.85, valign=c]{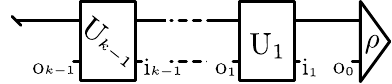}. \nn 
    \end{align}
Above, $\star$ denotes the Link product~\cite{chiribella_theoretical_2009,milz2020quantum}, corresponding to tensor contraction on the $\mc{H}_E$ subspace, and a tensor product on $\mc{H}_S$. We denote by $\mathcal{H}_{S(t_j)}^{\mathrm{i}}$ the input $\mathrm{i}$ to the process on the system Hilbert space, at time $t_j$. In the following, we will discuss mapping between different numbers of input/output pairs, or intervention times. We call an index pair $(\mathrm{i}_{j},\mathrm{o}_{j-1})$ a `time' or `slot' $j$, representing the Hilbert space {$  \mc{H}^\mathrm{i}_{S({t_j})} \otimes\mc{H}^\mathrm{o}_{S({t_{j-1}})}$.} Note that in the multitime (process tensor) picture, the system Hilbert space at different times are independent spaces, as is apparent from the graphical representation in Eq.~\eqref{eq:pt_link}. Through an abuse of notation, we will later take the times to be discrete, such that they label the $1^{\mathrm{st}}$, $2^{\mathrm{nd}}$, $3^{\mathrm{rd}}$, \dots, $k^{\mathrm{th}}$ intervention, $t_ j \in \{ 1,2,\dots,k\} $. 

$\ups_{k}$ in Eq.~\eqref{eq:pt_link} is the Choi state representation of the process tensor, i.e., a $(2k-1)-$body quantum state, where each `body' corresponds to either an input or output index at an intervention time. In other words, it is possible to show that it has all properties of a (supernormalized) density matrix
\begin{equation}
  \ups_{k} \ge 0, \quad \ups_{k}=\ups_{k}^\dag, \quad \tr[\ups_{k}]=d^{2k-1}.  \label{eq:normalization}   
\end{equation}
This can also be understood from the Choi-Jamio\l kowski isomorphism, i.e., the process tensor results from feeding in half of a maximally entangled state at each intervention time and collecting all of the outputs~\cite{processtensor}.
However, while process tensors are isomorphic to a quantum state, the converse is not true. A further sequence of affine constraints enforce the causal influence of interventions. These causality conditions are iteratively expressed by
    \begin{align}
             &\tr_{\mathrm{o}_j}[ \ups_{j} ] =   \id_{\mathrm{i}_j} \otimes \ups_{j-1}   \quad \forall \quad 1\leq j \leq k \nn \label{eq:causality}\\
            \text{with } \qquad  & \tr_{\mathrm{o}_0}[\ups_{0}] = 1,
    \end{align} 
i.e. tracing over a final output leg `commutes through' to the previous input. This ensures that an instrument applied at a given time cannot causally influence the statistics of any instrument preceding it in the past. Physically, discarding information at the latest available time step separates the preceding leg from the rest of the tensor, equivalent to the stochastic process property where the past should be unaffected by actions averaged across the future~\cite{milz2020quantum}. This isomorphism between a multitime process and a quantum state is a key insight, allowing for operationally meaningful notions of non-Markovianity~\cite{processtensor,processtensor2,milz2020quantum,white_non-markovian_2022} and genuinely quantum memory~\cite{Milz2021GME,taranto2023characterising}. It further provides a natural setting for studying optimal control~\cite{white_demonstration_2020,fux_efficient_2021, Berk2021resourcetheoriesof,berk2021extracting,butler_optimizing_2023}, as well as foundational questions regarding when large, chaotic systems become Markovian (simple)~\cite{FigueroaRomero_Modi_Pollock_2019,FigueroaRomero_Pollock_Modi_2021,Dowling2021,finitetime,Dowling2022}.

We stress that a particularly important instrument in the context of the process tensor formalism is the (trivial) identity map,
\begin{equation}
    \cup := \id \otimes \id = \includegraphics[scale=1.5, valign=c]{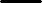},
\end{equation}
which simply corresponds to connecting input/output indices at a time step; a `do-nothing' operation. This is operationally distinct from a full measurement and averaging operation (tracing), in contrast to the classical case where these are equivalent~\cite{strasberg2019,Milz2020prx}. These facts forms the basis of the process tensor constituting the quantum generalization of stochastic processes~\cite{Milz2020kolmogorovextension,milz2020quantum}.


\subsection{Operational Non-Markovianity From a Process Tensor}
A key advantage of a process tensor representation of a dynamics is the ability to define non-Markovianity measures in an operational unambiguous way. One particularly elegant choice is the quantum mutual information $\eta$ between two channels of the process, with identity superoperators inserted everywhere else. 

The quantity $\eta$ is based on introducing a causal break on the $S$ space such that any remaining mutual information must be mediated by the environment $E$, leading to an instrument-independent measure {of} non-Markovianity. Explicitly, we find the marginal process tensor on two channels labeled $a$ and $b$, and find the quantum mutual information between them
\begin{equation}
    \begin{split}
        \eta(\ups; t_a, t_b)&:=S(\ups_{b,a}\| \ups_a \otimes \ups_b )\\
    &=S(\ups_{a})+S(\ups_{b})-S(\ups_{b,a}), \label{eq:NM1}
    \end{split}
\end{equation}
with
\begin{equation}
    \ups_{b,a} =\includegraphics[scale=0.85, valign=c]{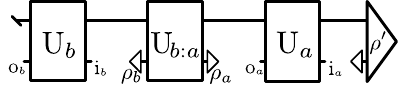}, \label{eq:NM2}
\end{equation} 
and where $\mathrm{U}_{b:a}:= \mathrm{U}_{b-1}\cdots \mathrm{U}_{a+1}$ and $\rho^\prime := \mathrm{U}_{a-1} \dots U_1 |\rho \rangle \! \rangle $. Then, $\ups_a := \tr_{b}[\ups_{b,a}]$ is the reduced state of $\ups_{b,a}$ on $\mathcal{H}_a = \mathcal{H}_{S(t_a)}^{\mathrm{i}}\otimes \mathcal{H}_{S(t_a)}^{\mathrm{o}}$, and $S(\sigma) := -\tr[\sigma \log \sigma ]$ is the von Neumann entropy. The independent preparation $| \rho_a \rrangle$ and measurement $\llangle \rho_b |$ together represent a `causal break', distinguishing temporal correlation within $\mc{H}_S$ from genuine non-Markovianity transferred through the environment $\mc{H}_E$. This measure Eq.~\eqref{eq:NM1} is operationally meaningful both as a distance to the closest Markovian process, and as the exponential scaling of the probability of mistaking your process as a Markovian one after $N$ experiments, $\mathbb{P} \sim \exp[- \eta N ]$~\cite{processtensor2,milz2020quantum}. It also agrees with the classical limit~\cite{processtensor2}.

The process tensor description we have detailed above is \textit{universal}. Namely, the process tensor can describe \textit{any} quantum process, Markovian or non-Markovian, with any amount of temporal correlations or memory (including `slow', polynomially decaying correlations), provided the dimension of the environment wire (Hilbert space) $d_E$ is allowed to be arbitrarily large. This can be seen clearly from the  graphical representation in Eq.~\eqref{eq:pt_link}. The process tensor there has a `tensor-train-like' internal structure, i.e. the form of an MPS~\cite{processtensor,Gu2018,milz2020quantum}, albeit it is a not pure state as it represents open quantum dynamics. When the environment dimension is fixed but the number of time steps {(denoted by $k$)} scales, the process tensor corresponds to a matrix product operator (MPO) with a bond dimension equal to $d_E$. It can be shown that such matrix product processes with finite $d_E$ have a \textit{finite} temporal correlation length, such that arbitrary correlations and indeed non-Markovianity decay exponentially; see App.~\ref{ap:linear}. This then raises the question, whether a process tensor can be rearranged into an alternative tensor network geometry? And in doing so can we efficiently model slowly decaying temporal correlations? This is the main task ahead in this work. We will see that a tree tensor network geometry provides an efficient and natural representation of processes with strong memory. To undertake this challenge, we first need to introduce the higher-order class of objects which map between process tensors.

\subsection{Quantum Superprocesses}
The basic building block of process trees, which we {define} in Sec.~\ref{sec:treeconstruct}, is a \textit{quantum superprocess} -- a causality and positivity-preserving linear map between two process tensors and/or (correlated) instruments~\cite{chiribella_theoretical_2009,chiribella_probabilistic_2010,Berk2021resourcetheoriesof}. In this work we consider the subclass of superprocesses which map between process tensors with a discrete number of time slots. As a map, such a superprocess acts by composition on the intervention slots of a process tensor and transforms it into a possibly different process tensor. In the Liouville representation, the action of such a superprocess is realized by contracting the superprocess represented as a tensor with the corresponding indices of the process tensor. We will now detail the minimal structure of this tensor.

Given a causal ordering across wires, a superprocess can be written as a sequence of CPTP maps acting on a combined input, output, and possibly ancilla Hilbert spaces. For example, a non-trivial superprocess between two input intervention slots {in a process tensor}, made of pairs of wires $(i,j)$ and $(k,l)$, and two output intervention slots on the {resulting process tensor}, $(i',j')$ and $(k',l')$, is
\begin{equation}
\mc{D}:=\includegraphics[scale=1.5, valign=c]{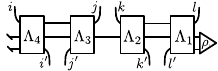}, \label{eq:D}
\end{equation}
where $\Lambda$'s are CPTP maps, {the wire at the top (bottom) corresponds to the input (output) spaces, the middle wire corresponds to an ancilla,} and $\rho$ representing some initial state of the ancilla plus the output space. The top open wires are time-ordered relative to the bottom open wires as $i \leftarrow i' \leftarrow j' \leftarrow j \leftarrow k \leftarrow k' \leftarrow l' \leftarrow l$. This then maps a process tensor to a generally different but physical process tensor on the same name of times, modified {time-locally} around the two times, $\mc{D} | \ups_{k} \rrangle = | \ups_{k}^\prime \rrangle$.

However, this is not the unique choice. Other two-time to two-time superprocesses are possible, corresponding to different \textit{relative} causal orders of the wires while keeping the causal orders of the inputs (outputs) fixed. For instance, the following superprocess,
\begin{equation}
\mc{F}=\includegraphics[scale=1.5, valign=c]{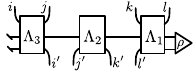}, \label{eq:D1}
\end{equation}
corresponds to the relative causal order $i \leftarrow \{ j,i' \} \leftarrow j' \leftarrow k' \leftarrow \{  k,l'\} \leftarrow l$. For example, for this causal ordering, state preparations on the index $k^\prime$ cannot influence the measurement statistics on the output space outgoing index $k$. This is in contrast to the causal ordering of Eq.~\eqref{eq:D}.

Superprocesses between an arbitrary number of {input and output} slots {(times)} can be constructed analogously, with each slot consisting of an in-going (to a CPTP map) wire and an out-going (from a CPTP map) wire such that the out-going wire is to the future of the in-going wire. Different possible superprocesses are enumerated by the different causal orders of all the wires (input and output) that fulfil this constraint.

In the following, we will be concerned with {superprocesses} which map from one-time to two-time. This will be the building block with which we construct the process tree, the main object of study in this work.

\section{Construction of process trees}\label{sec:treeconstruct}
We now formally construct a class of quantum processes with long-range temporal correlations, {as depicted in Fig.~\ref{fig:treeintro}}. Specifically, we first identify a general one-time to two-time temporally local superprocess, the building block from which we can iteratively construct process trees. Then, we next introduce a consistency condition. This leads both to an interpretation of the resultant superprocess as a change of scale transformation, as well as convenient numerical and analytic properties which we exploit in Section~\ref{sec:multitime}. 

{\subsection{Temporally Local Fine-graining of Processes} }\label{sec:y-brick}
The building blocks of a process tree are \textit{fine-graining} superprocesses, namely, causality-preserving maps that \textit{locally} (in time) transform a process tensor with $k$ intervention slots to a process tensor with $k+1$ intervention slots. A (time-)local superprocess acts on a single slot of a process tensor, say the $t^{\mathrm{th}}$ slot, without modifying the process to the past or future of the slot. In other words, the transformed process differs from the input process only at slot $t$.

Such a superprocess is a map from a single intervention slot to two slots,\textsuperscript{\footnote{One could also consider a fine-graining superprocess that maps an intervention slot to more than two slots. The process ansatz obtained from such higher-order fine-graining maps has all the same structural properties as the one presented in this paper. We chose one-to-two slot intervention maps in order to generate a simple binary tree tensor network.)}} $\mathcal{Y}: \mc{H}_{S_c}^{\otimes 2} \rightarrow \mc{H}_{S_f}^{\otimes 4}$, where $\mc{H}_{S_c}$ denotes the system space at the coarse level, and there are two copies for the input/output indices; see Section~\ref{sec:PT}. As remarked in the previous section, the causal ordering between the input and output slots and a choice of ancilla space almost entirely fixes the structure of the map (conditions (1)-(3) described below Eq.~\eqref{eq:D1}). In particular, we demand that all the information fed into the single, coarse input intervention slot should be able to affect the full measurement statistics of two fine interventions, and the latter can, in turn, influence the future of the process at the coarse level (via index $\mathrm{i}_c$ below). The most general map $\mathcal{Y}$ is parameterized by three CPTP maps $\Lambda_1, \Lambda_2,$ and $\Lambda_3$, an ancilla space (represented by the middle wire), and a preparation $\rho$ as depicted below: 
\begin{align}
        \mathcal{Y}^{\mathrm{i}_c \mathrm{o}_c}_{\mathrm{i}_{f} \mathrm{o}_{f} \mathrm{i}_{\bar{f}} \mathrm{o}_{\bar{f}} } &:=\includegraphics[scale=1.5, valign=c]{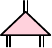} \nn\\
        &= \includegraphics[scale=1.5, valign=c]{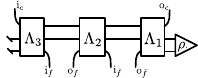}  \label{eq:y_brick}
\end{align}
where the upper (lower) indices $\mathrm{i}_c$ ($\mathrm{i}_{f}$) represent coarse (fine) temporal scales. For notational simplicity we will often drop the indices and discuss the full tensor (superprocess), where the input/output spaces will be clear graphically, and the full one-to-two superprocess is succinctly represented as a colored triangle. We prove explicitly in App.~\ref{ap:superprocess} that $\mathcal{Y}$ is indeed a superprocess that maps a process tensor with $k$ time steps to a process tensor with $k+1$ times, fulfilling the positivity and causality constraints, Eqs.~\eqref{eq:normalization}-\eqref{eq:causality}.

For a single time slot, a process tensor is physically equivalent to encoding all possible measurements and observables with respect to some reduced state $\rho_0$,
\begin{equation}
    \tree_0 = \includegraphics[scale=1.8, valign=c]{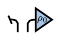} \label{eq:prepare}
\end{equation}
A process tree is then obtained by recursively fine-graining this simple case by applying a superprocess $\mc{Y}^{(j)}_s$, thus, refining the single intervention slot of the single-time measurement process into a \textit{multi-time} intervention space. That is, the first fine-graining map is denoted $\mc{Y}^{(1)}_1$, resulting in a process with two intervention slots, while the next fine-graining from $\mc{Y}^{(1)}_2$ and $\mc{Y}^{(2)}_2$ of each of these two slots results in a four-time process, and so on. After $N-1$ steps, we obtain a process tree with $2^N$ intervention slots. Fig.~\ref{fig:generaltree} illustrates a process tree with $2^{N=4} = 16$ intervention slots. Note that a process tree is explicitly dependent on the set $\{ \mc{Y}^{(j)}_s\}$ with indices $1\leq s \leq N$ and $1 \leq j \leq 2^{s-1}$, the initialization $\rho_0$, and the choice of local dimension $d_s$ at each scale $s$. We will often drop subscripts for notational simplicity. 

{The above temporal fine-graining transformations lead to tree tensor networks with} `baked-in' causality (Eq.~\eqref{eq:causality}).
Not only is the output of the final iteration a legitimate process tensor, but each iteration of the above fine-graining procedure produces a quantum process with twice as many intervention times as the previous step. Therefore, fine-graining produces a \textit{sequence} of quantum processes with doubling intervention slots, 
\begin{equation}\label{eq:sequence}
\tree_0 \rightarrow \tree_{1} \rightarrow ... \rightarrow \tree_N.
\end{equation}
Here, $\tree_{s}$ is the process at scale $s$, $\tree_0 $ is the initial prepare process at the coarsest scale, and $\tree_N$ is the process obtained at the finest time scale with $2^N$ intervention slots; see Fig.~\ref{fig:generaltree}.


\begin{figure}[t]
    \centering
    \includegraphics[width=7.5cm]{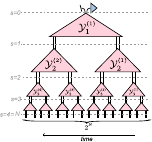}
    \captionsetup{justification=justified,singlelinecheck=false} 
    \caption{Recursive fine-graining of the initial single-time measurement process using $\mc{Y}$-type fine-graining maps generates the generic process tree, here shown for $N=4$ scales. The tensors $\{\mc{Y}^{(j)}_s\}$ are organized according to the $jth$ tensor at time scale $s$.}
    \label{fig:generaltree}
\end{figure}

\subsection{Consistency in Change of Scale} \label{sec:coarse}
{In order to interpret the different layers in Fig.~\ref{fig:generaltree} as temporal `scales', we will now impose a consistency condition on the $\mc{Y}-$bricks from Eq.~\eqref{eq:y_brick}. This will allow us to interpret the hierarchy of process trees in Eq.~\eqref{eq:sequence} as different temporal scales of the same physical process }

To understand why extra structure is desirable, consider a height $N=4$ process tree as shown in Fig.~\ref{fig:generaltree}, and consider the application of {(a sequence of)} instruments at a coarser scale, e.g. the scale $s=2$. There are two natural choices in how to extract such a correlation function from a $N=4$ height process tree, one could either: (i) `chop' the shown process tree and compute the correlators on the coarse process tree consisting of only the three superprocesses $\{ \mc{Y}^{(1)}_1,\mc{Y}^{(1)}_2,\mc{Y}^{(2)}_2\}$, on layers $s=0,1,2$. This means we essentially ignore the finer scales $s=3,4$; or (ii) insert `do-nothing' operations on the finest scale ($s=4$), contract the process tree to $s=2$, and then insert appropriate instruments; this will generally depend on the full set $\{ \mc{Y}^{(j)}_s \}_{s=1}^4$. These two choices are equally valid physically, but are generally different. We choose the following condition to impose that these situations (i) and (ii) are equivalent,   
\begin{align}
        &\mc{Y}^{\mathrm{T}} | \cup \rrangle_f \! | \cup \rrangle_{\bar{f}} = | \cup \rrangle_c \nn \\ 
        \iff & \sum_{\mathrm{i}_{f} \mathrm{o}_{f} \mathrm{i}_{\bar{f}} \mathrm{o}_{\bar{f}}} (\mathcal{Y}^{\mathrm{T}})^{\mathrm{i}_c \mathrm{o}_c}_{\mathrm{i}_{f} \mathrm{o}_{f} \mathrm{i}_{\bar{f}} \mathrm{o}_{\bar{f}} } \delta_{\mathrm{i}_{f} \mathrm{o}_{f}} \delta_{ \mathrm{i}_{\bar{f}} \mathrm{o}_{\bar{f}}} =  \delta^{ \mathrm{i}_{c} \mathrm{o}_{c}} \label{eq:isometry} \\
        \iff & \includegraphics[scale=1.5, valign=c]{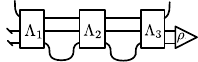} = \includegraphics[scale=1.5, valign=c]{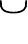}. \nn
\end{align}
For clarity here we have written the same condition in three equivalent representations, respectively: superoperator, index, and graphical. We call this the \textit{scale consistency condition}. {The superscript `$\mathrm{T}$' means transpose, as by convention we have take the $\mc{Y}$ to act from top-to-bottom as a `fine-graining' map. Such notation is conveniently not necessary in the graphical representation.} We call process trees composed of superprocesses satisfying Eq.~\eqref{eq:isometry} $\mc{W}-$type, and these will be the focus of the rest of the paper. 

We stress that imposing the scale consistency condition is a choice, and $\mc{Y}-$type process trees composed of superprocesses which do not satisfy Eq.~\eqref{eq:isometry} may describe interesting physical phenomena. Further, we could define a scale consistency condition with respect to different instruments $\mathrm{V}$, such that $\mc{Y}^{\mathrm{T}}|{\mathrm{V}}\rrangle |{\mathrm{V}}\rrangle  = |{\mathrm{V}}\rrangle$. If $\mathrm{V}$ corresponds to a unitary map, then the resultant process tree is equal to a $\mc{W}-$type process tree, up to (temporally) local transformations only at the coarsest and finest scales, $s=0$ and $s=N$ respectively; see App.~\ref{ap:y-type}. {It remains to be seen whether different scale conditions could lead to relevant models; $\mc{Y}-$type process trees serve as a general base from which to study this.  }

Eq.~\eqref{eq:isometry} is, however, a natural choice.\textsuperscript{\footnote{We do not claim that this choice will capture all physics. There may be superprocesses that do not preserve the identity operations, that capture interesting physics.}} Intuitively, we are imposing that a `do nothing' operation at a fine scale corresponds to `do nothing' at a coarse scale. This leads to a number of nice properties which we investigate in the remainder of this paper. In particular, we see a kind of causal cone structure in the scale direction, for temporally local quantities. This is represented diagrammatically in Fig.~\ref{fig:causalcone1}, where the blue colored region represents all the tensors that need contracting in order to compute a single-time expectation value. This means that local quantities will depend only on a restricted number of tensors in the process tree, linear in the height of $N$ of the tree, compared to the general $\mc{Y}-$class tree which involves at worst contractions of all $O(2^N)$ bricks (the entire tensor network). This leads to computations using the process tree that are computationally efficient, and even analytically tractable in certain asymptotic regimes (see Thm.~\ref{thm:polyDecay}). We will study this next in Sec.~\ref{sec:multitime}. 

We can intuitively understand the consequences of Eq.~\eqref{eq:isometry} by examining the time-resolved description of a process tree. Here, tensors at different time scales in the network, indicated by the dashed lines in Fig.~\ref{fig:generaltree}, capture properties of the process at different time scales. More specifically, if the system is intervened on time scales only longer than a scale $s$, then all the relevant properties -- such as multi-time correlation functions -- of the process can be calculated entirely from the tensors located above scale $s$ in the network. The remaining tensors, corresponding to shorter time scales, can be then discarded.

We will explicitly incorporate the constraint Eq.~\eqref{eq:isometry} by choosing $\Lambda_1 = U_2, \Lambda_2 = U_2^\dagger U_1$, and $\Lambda_3 = U_1^\dagger$ in Eq.~\eqref{eq:y_brick}, where $U_2$ and $U_1$ are unitary maps. Explicitly, 
\begin{align}
        \mathcal{W} &=\includegraphics[scale=1.5, valign=c]{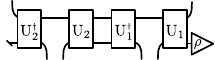}  \label{eq:w_brick} =:\includegraphics[scale=1.5, valign=c]{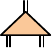}
\end{align}
 It is readily checked that the superprocess $\mathcal{W}$ satisfies Eq.~\eqref{eq:isometry}. The map $\mc{W}$ can be viewed as a fine-graining transformation of processes (given its property as a superprocess),
\begin{equation}
\includegraphics[scale=3.25, valign=c]{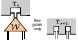}, \label{eq:fineone}
\end{equation}
while the transpose map $\mc{W}^{T}$ acts as a coarse-graining transformation between scales on instruments (given the scale consistency condition Eq.~\eqref{eq:isometry}):
\begin{equation}
\includegraphics[scale=3.25, valign=c]{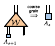} \label{eq:coarseone}
\end{equation}
{Here, through abuse of (graphical) notation, we have taken the gray boxes to represent process tensors that may be any size, as defined in Eq.~\eqref{eq:pt_link}. Only the one input and two output time slots of this process are relevant here, as these superprocesses are (time) local. $s$ refers to the temporal scale (as in Fig.~\ref{fig:generaltree}).} Note that the transpose map $\mc{W}^T$ coarse-grains instruments, not processes, {with $|A_s\rangle \! \rangle = \mathcal{W}^{\mathrm{T}} |A_{s+1}\rangle \! \rangle |\cup \rangle \! \rangle$}. Ostensibly, a coarse-graining map for processes must be a superprocess (causality preserving) that should \textit{invert} the action of $\mathcal{W}$ as a fine-graining map on a process. While it is unlikely that such an inverse exists for all $\mathcal{W}$, it is an open question what properties of $U_2$ and $U_1$ in Eq.~\eqref{eq:w_brick} (or of a more general $\mc{Y}$) imply the existence of an inverse. Unitaries satisfying the so-called `dual-unitary property' appear to be a promising candidate~\cite{Bertini2019exact}, but we leave a detailed study of this question to a future work. 

In the remainder of this paper, we focus exclusively on process trees composed only of the $\mathcal{W}$-type superprocesses defined through Eq.~\eqref{eq:w_brick}, and use `process tree' to refer hereon only to such processes. More details on process trees without the scale consistency condition ($\mc{Y}-$type) can be found in App.~\ref{ap:y-type}.

\section{Correlation Functions of Process Trees}\label{sec:multitime}
{In this section, we describe how $k$-time correlation functions of a process tree can be computed efficiently, and showcase their behavior. 
In particular, we show both numerically and analytically that temporal correlations of a generic process tree decay with a characteristic power-law. }

\begin{figure}[t]
    \centering
    \includegraphics[width=\linewidth]{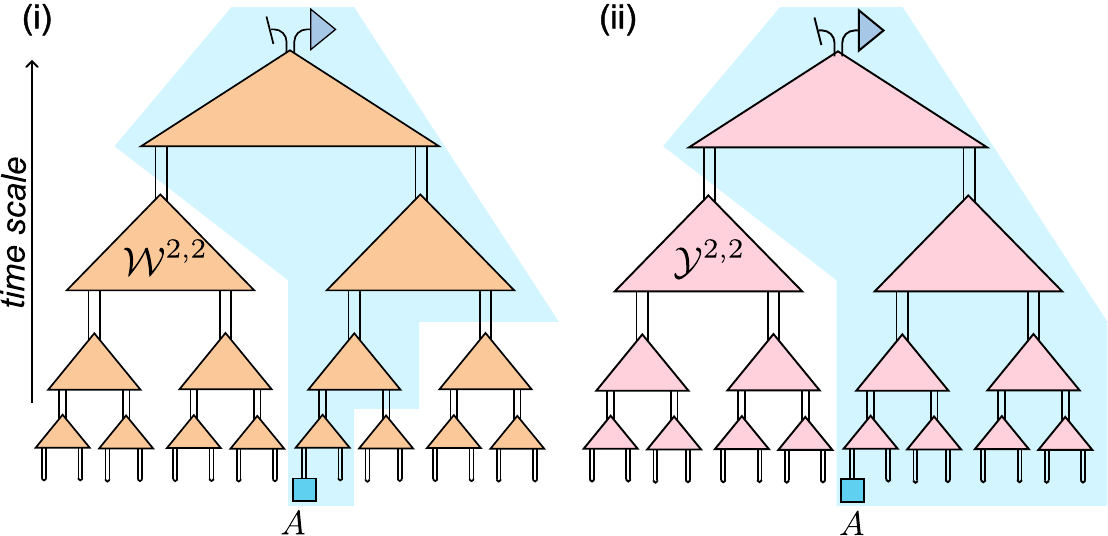} \captionsetup{justification=justified,singlelinecheck=false} 
    \caption{(i) The scale-causal cone  (shaded blue) of a single intervention slot in a process tree composed from $\mc{W}$-type maps, illustrating the emergent causal structure also in the scale direction. Only tensors inside the causal cone influence the intervention slot (for instance, the expectation value of an intervention $A$ applied on that slot). The causal cone is comprised of exactly one tensor from each scale. (ii) An intervention slot in a generic process tree, composed from $\mc{Y}$-type maps is influenced by all past tensors at all scales. }
    \label{fig:causalcone1}
\end{figure}

\subsection{One-time expectation values and emergent causal structure along scale} \label{sec:one-time}
Fig.~\ref{fig:causalcone1} illustrates the tensor network contraction that evaluates the expectation value of a single instrument $A$. We see that the instrument is contracted with one slot of the tree, while the remaining slots are contracted with the Identity. Thanks to the scale consistency condition fulfilled by $\mathcal{W}$-type maps,  Eq.~\eqref{eq:isometry}, in this case the total contraction that evaluates the expectation value simplifies significantly. For instance, the contraction depicted in Fig.~\ref{fig:causalcone1}(i) reduces to 
\begin{align}
    \braket{A}_{\tree}&=\includegraphics[scale=1.2, valign=c]{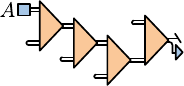} \label{eq:1-time}\\
    &= \llangle A |   \mc{W}^{(4)}_{4,L} \mc{W}^{(2)}_{3,L} \mc{W}^{(1)}_{2,L} \mc{W}^{(1)}_{1,R} |{\rho}\rrangle |{\id}\rrangle,
\end{align}
where $\llangle A | = \llangle{\phi^+} | \mc{A}^\dg \otimes \id$ is the dual Choi state of some instrument $A$, and 
\begin{equation}
\mathcal{W}_R = \includegraphics[scale=1.2, valign=c]{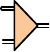}, \, \mathcal{W}_L = \includegraphics[scale=1.2, valign=c]{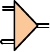}.\label{eq:Tdefa}
\end{equation}
Note that while graphically we have depicted tensors at different scales with congruent triangles, each superprocess $\mc{W}^{(j)}_s$ could be different in general.

From Eq.~\eqref{eq:1-time} we can readily estimate the computational cost of computing single-time correlation function in a process tree. In particular, $\braket{A}$ depends only on a linear subset of tensors in the process tree; a single superprocess $\mc{W}$ per `scale' $s$. Therefore, for a height $N$ process tree $\tree_N$ it reduces to a contraction of $N$ tensors $\mc{W}^{(j)}_s$. Assuming that each superprocess $\mc{W}^{(j)}_s$ has the same `bond dimension' $d$ (i.e. same input and output space dimensions between scales), $\braket{A}$ requires $N d^6$ operations. See App.~\ref{ap:tn} for further details about estimating the computational cost of tensor network contractions, and App.~\ref{ap:numerics} for further details on contracting the tree.

This contraction behavior can be interpreted as an emergent a causal structure in the bulk of the tensor network. We call the set of tensors that influence an intervention slot the \textit{scale causal cone} of the slot, depicted in Fig.~\ref{fig:causalcone1}. For the $\mc{W}-$brick (satisfying Eq.~\eqref{eq:isometry}), the scale-causal cone of any intervention slot consists of exactly one tensor at each time scale. Note that this emergent causal structure in the scale direction is distinct from the causal structure in the time direction. The latter results from the process tree fulfilling the causality constraints, Eq.~\eqref{eq:causality}. Recall that the most general process trees -- those composed of $\mathcal{Y}$-maps -- fulfil causality constraints along the direction of the intervention time, such that future dynamics have no influence on past interventions, but they do not have this kind of scale causal structure in the scale direction. That is, in a process tree made from $\mc{Y}$-type tensors, an intervention slot is influenced by all the past tensors at \textit{all} scales, as illustrated in Fig.~\ref{fig:causalcone1}(ii). This is detailed further in App.~\ref{ap:y-type}.

The scale causal cone that manifests in process trees can be understood as a temporal analog of the causal cone in spatial tree tensor networks and MERA representations of quantum many-body states. In the spatial case, the causal cone structure results from the local isometric property of the tensors, namely, the product of each tree or MERA tensor component with its hermitian adjoint equates to the Identity. In the present case of process trees, the causal cone structure results instead from enforcing the scale consistency condition, Eq.~\eqref{eq:isometry}. Spatial tree tensor networks and MERA, equipped with their respective causal cone structures, have been interpreted as encoding an emergent holographic two-dimensional anti de-Sitter geometry \cite{swingle_entanglement_2012,Bao2017}. We remark that the emergence of the causal structure in process trees might also lead to an analogous holographic description of quantum processes, but further exploration of this feature of process trees is beyond the scope of the current paper.

\subsection{Two-time correlators}
Next, let us analyze how two-time correlators generally scale with the duration $\Delta t = t'-t$ between the two interventions. {Fig.~\ref{fig:causalcone2} shows the tensor network contraction that equates to the correlator between two instruments $A$ and $A'$ applied on time slots $t$ and $t' > t$, respectively. Notice that the scale-causal cones of $A$ and $A'$ overlap beyond some scale $s$. This can be implemented efficiently numerically, leading to a computational cost $\sim\mc{O}(\mbox{log}_2(N))$ similar to the single-time case (see App.~\ref{ap:numerics}). }
It is not apparent that enforcing causality and the coarse-graining constraints at the level of individual tensors, as in a $\mc{W}$ process tree, preserves the polynomial decay (critical behavior) of correlations that originates in the tree geometry of the network. However, we find that these constraints are, in fact, compatible with polynomially-decaying correlations. In Fig.~\ref{fig:numerics}(ii), we plot the averaged two-time correlator in a height $N=8$ uniform process tree. A uniform process tree is composed of $\mc{W}^{(j)}_s \equiv \mathcal{W}$ for all $s,j$, where for the numerical results $\mathcal{W}$ is composed of randomly sampled (according to the Haar measure) unitary maps $U_2$ and $U_1$ in Eq.~\eqref{eq:w_brick}. For each value of separation, $\Delta t = t' - t$, we averaged the correlators of fixed instruments $A$ and $A'$ inserted at all possible intervention slots $t$ and $t' > t$ such that $t' - t = \Delta t$. These results demonstrate that temporal correlations in process trees decay polynomially, and moreover, this feature is generic. Namely, this behavior is a structural property of the class of processes (specifically, the tree geometry of the network \cite{evenbly_tensor_2011}) and does not require fine-tuning of the tensors.

\begin{figure}[t]
    \centering
    \includegraphics[width=8cm]{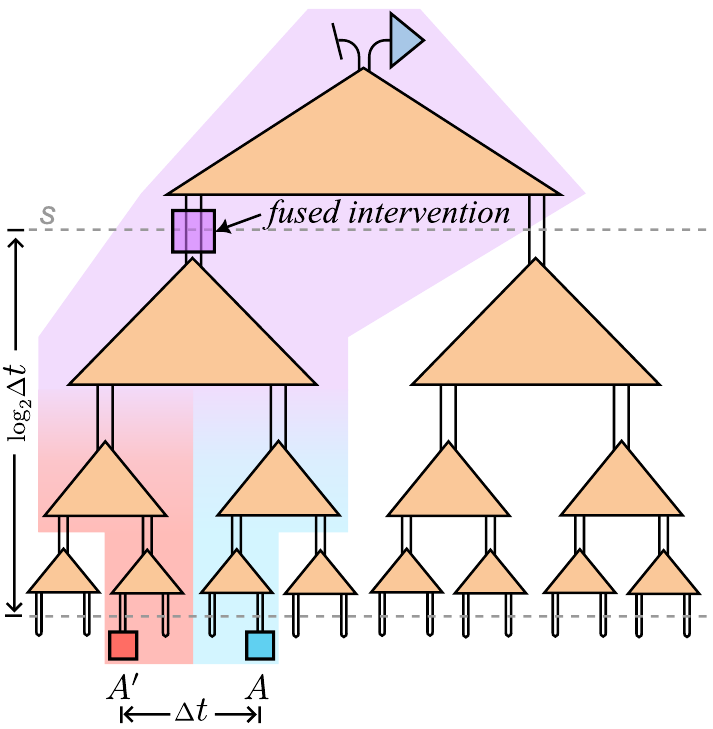} \captionsetup{justification=justified,singlelinecheck=false} 
    \caption{A diagrammatic representation of the contractions involved in computing two-point correlations in a process tree. The joint scale-causal cone of two interventions $A$ and $A'$ applied on slots $t$ and $t^\prime >t$, respectively. The correlator of the two interventions depends only on the tensors inside the causal cone. The two-time correlator of $A$ and $A'$ equates to a one-time expectation value at a scale $s$ where $A$ and $A'$ fuse together after coarse-graining $s$ times. On average across all sites $\{ t, t^\prime \} $ such that $t^\prime - t = \Delta t$, the number of required coarse-graining moves until the operators `fuse' is $s \sim \log_2(\Delta t) $. See App.~\ref{ap:numerics} for a more detailed explanation of the efficient computation of correlation functions. }
    \label{fig:causalcone2}
\end{figure}

\begin{figure*}[t]
    \centering
    \includegraphics[width=\linewidth]{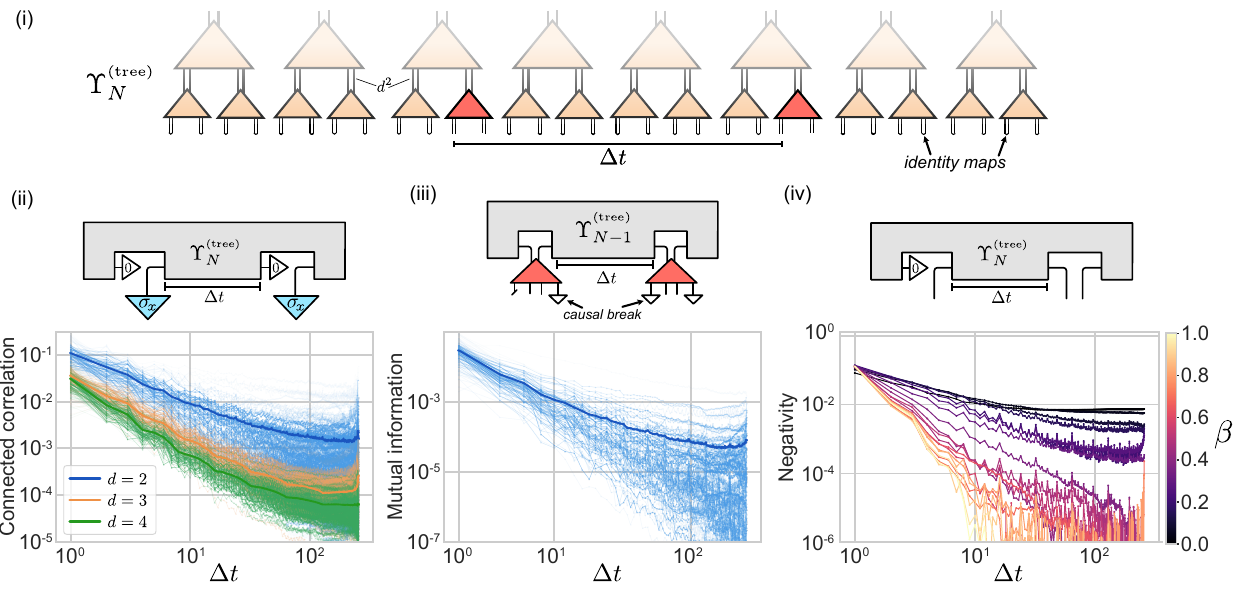}
    \caption{Numerical Results for general process trees. (i) All quantities shown in the plots can be derived from a process tree with identity maps $\cup$ inserted everywhere, except for the two $\mc{W}$ tensors which are $\Delta t$ apart. (ii) Connected two-point correlations for Haar random samplings of the unitary maps $U_1$ and $U_2$, and $\rho$ from Eq.~\eqref{eq:w_brick}. The process tree exhibits polynomially decaying correlation between two instruments applied on the open slots at scale $s=0$. The simulation was performed for a tree of depth $N=8$, which contains $2^8 = 256$ intervention times. $d$ is the dimension of system at each scale; equivalently $d^2$ is the `bond dimension' of the tensor network representation of the process tree. We chose the Hermitian observables $A = A^\prime $ to be a Pauli $\sigma_x$ measurement on the output index, and an independent preparation of the $\ket{0}$ state on the input index for each intervention time. The first Gell-Mann matrix and $\sigma_x \otimes \sigma_x$ was instead measured for $d=3$ and $d=4$ respectively. {Such observables take expectation values $|\braket{A}|\leq 1$, and so the presented data represents a significant range}. Each faded line is a single run of the simulation {for $100$ runs per plot,} and the bolded line is the average of the runs. (iii) Quantum mutual information (Eq.\eqref{eq:NM1}) measure for non-Markovianity between reduced channels in a homogeneous process tree, produced for random, homogeneous process trees as in the two-point correlation results. We see power-law decay of non-Markovianity, indicating that memory follows the same trend as general correlations for particular instruments. (iv) Negativity witness (Eq.~\eqref{eq:negativity}) of entanglement in time for the case of process trees with unitary maps $U_2$ and $U_1$ generated by the Hamiltonian in Eq.~\eqref{eq:heisenberg}. We see that negativity, hence entanglement in time, decrease as we increase randomness in the process. As $\beta$ increase, the time $\Delta t$ beyond which entanglement in time is zero decreases, but for all cases we see power-law trend.}
    \label{fig:numerics}
\end{figure*}

Furthermore, in App.~\ref{ap:proof}, we prove the following statement pertaining to the asymptotic scaling of connected correlators in a uniform process tree. 
\begin{restatable}{thm}{polyDecay} \label{thm:polyDecay}
    Consider a uniform process tree $\tree_N$ of $2^N$ intervention times, composed from a `generic' fine-graining superprocess $\mathcal{W}$ defined in Eq.~\eqref{eq:w_brick}, and two instruments $A_t$ and ${A'}_{t'}$ applied on slots $t=t_0$ and $t^\prime=t_n >t_0$ with $\Delta t = t_n - t_0 = 2^n-1$ ($n$ is a positive integer). Then in the asymptotic limit $N,n \rightarrow \infty$,
    \begin{equation} \label{eq:poly_exact}
        |\braket{A_{t_n}^\prime A_{t_0}}_{\tree_N} - \braket{A_{t_n}^\prime }_{\tree_N}\!\braket{A_{t_0}}_{\tree_N}  | {\sim} \Delta t^{-\alpha},
    \end{equation} 
    where $\alpha >0$. The sufficient conditions on the superprocess $\mc{W}$ is defined in terms of spectral properties of a derived transfer matrix, given in Eq.~\eqref{eq:Tdefa}.   
\end{restatable}
{We note that the spectral conditions on $\mc{W}$ in the above theorem are highly nonrestrictive; numerically always being satisfied (and so we expect these to hold almost surely).} We limit to the case where $\Delta t = 2^n-1$ because it leads to a symmetric, repeated structure amendable to an analytically tractable proof without averaging. However, as exhibited in Fig.~\ref{fig:numerics} (ii), for randomly sampled $\mc{W}-$bricks one expects polynomial decay of correlations on average for all $\Delta t$. A process tree may of course also have an inhomogeneous structure, while still being $\mc{W}-$type (i.e. satisfying Eq.~\eqref{eq:isometry}). We expect such a process tensor to also exhibit polynomially decaying correlations, and this more general setting will be relevant later when studying the spin-boson model in the context of a process tree ansatz; see Sec.~\ref{sec:grug}.

Here, we summarize the main argument underlying the proof. Applying the scale consistency condition, Eq.~\eqref{eq:isometry}, the unconnected correlator simplifies to:
\begin{align}
        \braket{A_{t_n}^\prime A_{t_0}}&= \!\includegraphics[scale=1.2, valign=c]{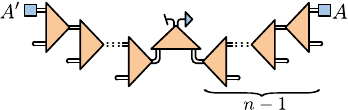}, \nn  \\
&=\llangle A^\prime |(\mc{W}_L)^{n-1} \mc{W}^{(\mathrm{b})} (\mc{W}_R^{\mathrm{T}})^{n-1} |A\rrangle ,\label{eq:corr_function}
\end{align}
where we choose $t=0$ and $t^\prime = 2^n-1$, in order to get a repeating structure. We see that two matrices $\mathcal{W}_{R}$ and $\mathcal{W}_L$ appear repeatedly, $n-1$ times, on the right and left branches of Eq.~\eqref{eq:corr_function} respectively. These matrices are what we call `right' and `left' contractions from Eq.~\eqref{eq:Tdefa}, and are two of the fundamental steps in computing any correlation functions in a process tree (numerically or otherwise); see also App.~\ref{ap:numerics}. Here, $\mathcal{W}_{R/L}$ are CP maps, but need not be unitary, unital, or TP in general. However, due to the scale consistency condition Eq.~\eqref{eq:isometry}, $\mathcal{W}_{R/L}$ has the Identity supermap $| \cup \rrangle$ as an eigenvector with the largest eigenvalue equal to one. Therefore, the spectral radius of $\mathcal{W}_R$ is equal to one; that is, all eigenvalues of $\mathcal{W}_R$ lie on the unit disk in the complex plane. Numerical evidence convincingly implies that this largest eigenvalue is unique for generic $\mc{W}$. From this, one can show that for large $n$,
\begin{equation}\label{eq:temppp}
    \braket{A_{n} A'{}_{n'}}_{\tree_N} \sim \braket{A_{n}}_{\tree_N}\braket{A'{}_{n'}}_{\tree_N} + O(|\lambda^{L}_2 \lambda^{R}_2|^{n})
\end{equation}
where $\lambda^{L}_2, \lambda^{R}_2 < 1$ are the second largest eigenvalue of $\mathcal{W}_{L}, \mathcal{W}_R$, respectively. To prove this we need to use modified Quantum Frobenius Perron Theorem~\cite{Bhatia2007,Wolf_undated-dn}. Then since $\Delta t = 2^n -1$, we re-write Eq.~\eqref{eq:temppp} as
\begin{equation}
    \braket{A_{n} A'{}_{n'}} \sim \braket{A_{n}}\braket{A'{}_{n'}} + O(\Delta t^{\mbox{\tiny log}(\lambda_2)}),
\end{equation}
where $\lambda_2 = |\lambda^{L}_2 \lambda^{R}_2| < 1$.
Therefore, the connected correlator, Eq.~\eqref{eq:poly_exact}, decays as $\Delta n^{-\alpha}$, with $\alpha = |\mbox{log}(\lambda_2)|$.

One can make similar arguments to examine the asymptotic behavior of higher point correlation  functions. Generally such correlations will decay polynomially with time between each nearest intervention. That is, the above result can be generalized to show that $\braket{A_{t_k }  \dots A_{t_2 } A_{t_1 }  A_{t_0}}_{\tree} \sim  (\Delta t_k)^{-\alpha_k} \dots (\Delta t_2)^{-\alpha_2}  (\Delta t_1)^{-\alpha_1} $, where $t_k < \dots < t_1 < t_0$ and $\Delta t_i = t_{i+i} - t_i$. Computationally, for $\mc{W}-$class process trees, $k-$time correlations reduce to a concatenation of $O(k N)$ right/left contraction moves, as described in the computation of one-time correlations in Sec.~\ref{sec:one-time}, and of $k$ `fusion moves' as in the computation of two-time correlations above. We detail this in more detail in App.~\ref{ap:numerics}.

While in this section, and throughout this work, we have examined temporal scaling of observables measured on a single time scale $s=N$, the process tree allows for \emph{cross-scale} interventions. What this means is that we could find the correlation between, for example, an instrument at the $t^{\text{th}}_s$ time slot at a coarse scale $s$, together with an instrument at some finer scale $s^\prime$ at slot $t^{\prime}_{s^\prime}$. Physically, one could also in principle measure a quantum system at some coarse time scale, and apply unitary control mechanisms at a finer scale. If one has a process tree representation of a process, and the different `scales' $s$ in fact correspond to physical temporal scales, then by construction such operations are easily accessible. This is because we have built this model out of locally causality preserving bricks, rather than imposing causality preservation globally. More technically, process trees are constructed from local one-to-two time superprocesses, rather than the much more general situation of a one-to-many time global superprocess. This feature of process trees could be relevant to a range of physical applications where vastly different time scales appear naturally, such as in quantum computing setups where the time scales for single-qubit operations, two-qubit operations, and measurements are vastly different. Such a device with many qubits may possess a complex (power law) noise profile. It would be interesting to explore this further in future work.


\section{Nature of Long Range Correlations}
\label{sec:NM-cor}
We have seen that the process tree generically exhibits polynomial decaying multitime correlations, suggestive of its utility in describing complex physical dynamics. Before analyzing the process tree structure of the spin-boson model, we will first investigate the nature of these strong temporal correlations. In particular, whether these correlations stem from memory effects, and if they can be genuinely quantum, i.e. involve entanglement in time. 

\subsection{Non-Markovianity}\label{sec:nm}


In this section, we analyze non-Markovianity in process trees.
{Recall the operational measure of non-Markovianity as described in Sec.~\ref{sec:PT}: quantum mutual information $\eta(\ups_N;a,b)$ between local channels in a process tensor (Eq.~\eqref{eq:NM1}).} This quantifies how much information in transferred through the environment, in terms of how well an optimal Markovian process models all accessible measurement statistics of a process~\cite{processtensor2}. We numerically plot $\eta(\ups_N;a,b)$ in Fig.~\ref{fig:numerics}(iii), for a uniform $\mc{W}-$class process tree. We can see essentially complete agreement between this and two-point correlation behavior in Fig.\ref{fig:numerics} (ii). This means that the characteristic behavior for multitime correlations that we saw in the last section stems primarily from non-Markovianity. This is by construction, which we now explain.

\begin{figure*}[t]
    \centering
    \includegraphics[width=\linewidth]{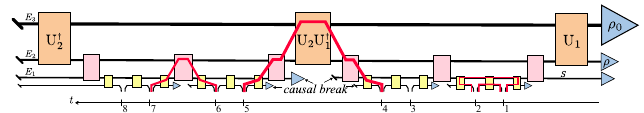}
    \captionsetup{justification=justified,singlelinecheck=false} 
    \caption{The circuit representation of a $3-$level process tree $\tree_3$, expanding out the explicit expressions for the constituent $\mc{W}-$bricks Eq.~\eqref{eq:w_brick}. The environment wires $E_s$ are organized in time scales $s$ with higher time scales are depicted with thicker lines, and the corresponding tensor colored according to scale. The thin wires at the bottom of the tensor network represent the system $S$ on which interventions may be applied. Shown also in red is the possible flow of correlations between three different pairs of neighboring intervention slots, $t=\{1,2 \},\,\{4,5 \},\, \& \,\{6,7 \}$. By construction, a process tree includes causal breaks at different time scales, forcing correlations to route through higher time scales. Note that correlations only flow only from past to future.   }
    \label{fig:circuit}
\end{figure*}

To explain why non-Markovian effects dominate in the process tree, it is helpful to consider how one would compute $\eta(\ups_N;a,b)$ in practice. Through quantum process tomography~\cite{white_non-markovian_2022}, one can reconstruct a process tensor Choi state through a complete basis of measurement and preparations. From this Choi state, $\eta(\ups_N;a,b)$ is directly computable. From Eq.~\eqref{eq:NM1}, for the process tree tomographic experiments will have the form of the multitime correlation
\begin{equation}
    \begin{split}
            \llangle \tree_{b,a} | A \otimes A^\prime \rrangle &=\includegraphics[scale=2, valign=c]{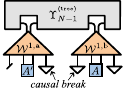}, \\
    &=\includegraphics[scale=2, valign=c]{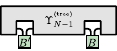},
    \end{split}
\end{equation}
where the gray comb represents the rest of a process tree, as in Fig.~\ref{fig:generaltree}, and $A, A^\prime$ are some tomographically selected instruments of the channels (usually a prepare and measure pair). The other projections (measurements/preparations) which are represented as uncolored triangles are unimportant, and can be chosen freely as long as they are fixed across experiments and are uncorrelated. The important thing here is that at the finest scale $s=0$, the quantum mutual information includes a causal break. This prevents any (quantum) information flow through the system wire, such that any non-zero mutual information must be due to correlations mediated solely through the environment (all other wires). We can see in the second line that the two-time instrument $A$ ($A^\prime$) coarse grains to an effective \textit{one-time} instrument $B$ ($B
^\prime$). $B$ will generally not itself involve any causal break. The remaining process tree $\tree_{N-1}$, with one less layer, will exhibit polynomially decaying correlations from the tree-like geometry of $\tree_{N-1}$. The self similarity of the tree implies that one sees long-range correlations at all scales, and therefore any causal break in a non-Markovianity measure will not interrupt this. 


In order to understand the non-Markovian structure of process trees, in Fig.~\ref{fig:circuit} we re-express a depth$-$3 process tree in terms of its component maps. As detailed in Sec.~\ref{sec:treeconstruct} and Fig.~\ref{fig:generaltree}, this consists of fine-graining a one time measurement process by iteratively applying the superprocess $\mathcal{W}$. Graphically, acting with $\mc{W}$ introduces a new system wire, with the old system wire from the previous (coarse) layer now mediating memory effects. In this way, we can see that fine-graining iteratively produces a non-Markovian process with memory distributed across all time scales. In particular, each additional layer introduces an additional single line of memory, and so memory dimension increases \textit{exponentially} with the number of time scales. Here, we have labeled the wires that carry environmental degrees of freedom as $E_s$. Temporal correlations in the process are then mediated through the various wires, and examples of memory traversing different layers are shown in colored lines.

Consider now the temporal correlations between neighboring intervention slots $t$ and $t+1$. Each of these next-neighbor sites are included in the averaging of the results from Fig.~\ref{fig:numerics}. Fig.~\ref{fig:circuit} illustrates that for different intervention times $t$, correlations between neighboring slots could be carried by wires at different time scales. Shown in the figure are three different pairs of neighboring slots, corresponding $t=\{1,2 \},\,\{4,5 \},\, \& \,\{6,7 \}$, where correlations are mediated only through scales $s=0 \, \& \, 1$, $s=3$, and $s=2$ respectively (correlation `paths' are shown in red). Note that the mediation of correlations is also limited by causality -- correlations cannot flow from the future back into the past. The latter two paths, mediate strictly non-Markovian correlations. This is because, by construction the process tree contains a hierarchy of causal breaks that force correlations to be routed through wires at longer time scales. It is these built-in causal breaks that organize the total memory in the process into different time-scales. On the other hand, the interventions at times corresponding to $t=\{1,2 \}$ are correlated through environmental and system paths, implying that the correlations between these two slots will generally be a mix of Markovian and non-Markovian. 

{We note that the above analysis of correlation (Fig.~\ref{fig:circuit}) is for a single copy of a process tree, treating the individual wires (Hilbert spaces) as being physical environmental degrees of freedom. Later, when we find the tree representation of physical processes from variational model fitting, the analysis of this section is difficult to apply. In this case, the meaning of the different `scale' levels $s$ is an open question. On the other hand, the process tree which would result from the method in Section~\ref{sec:plaquette} retains the circuit structure. Depending on the model, in this case each scale corresponds to different environmental modes or systems, the the details of Fig.~\ref{fig:circuit} inform us of the correlation structure therein.}

\subsection{Entanglement in Time}\label{sec:entanglement}
We have seen that in generic process trees correlations and the degree of memory generically decay polynomially. It therefore seems to be a good model for long-range memory dynamics. But are these correlations classical or quantum? The measures of correlations that we have looked at indiscriminately quantify correlations in the Choi state of the process. Since the Choi state is generally a mixed state, the corresponding correlations could be entirely classical. That is, the memory in the process tree could be classical, such that all multitime correlations could be producible from a classical register~\cite{Milz2021GME,taranto2023characterising}, or could even be modeled by an entirely classical stochastic process~\cite{strasberg2019,Milz2020prx}. 

We would therefore like to test whether it is possible to have entanglement in time in a process tree. From an operational viewpoint, entanglement in time quantifies the extent to which the multi-time correlations in the process require genuinely quantum memory and cannot be reproduced by a classical memory~\cite{taranto2023characterising}. Mathematically, this is equivalent to measuring bipartite entanglement in the Choi state of the process tensor, that is, entanglement across two intervention slots. Note that `entanglement in time' should not be confused with the (unfortunately named) quantity of `temporal entanglement' in influence matrices~\cite{Lerose2021,Lerose2021prb,Bertini2023}, an object related to the process tensor. Despite its name, `temporal entanglement' is a measure of total correlation (both classical and quantum), more akin to quantum mutual information $\eta$ (Eq.~\eqref{eq:NM1}) rather than entanglement in time.

While there exists a range of inequivalent monotones for entanglement in mixed quantum states, we here look at the entanglement negativity of the Choi state of the process tree~\cite{Wilde_2011},
\begin{equation}
    \mc{E}(\tree_{t t^\prime}) := \frac{ \| (\tree_{t t^\prime})^{{\mathrm{T}_t}}\|_1-1}{2}, \label{eq:negativity}
\end{equation}
where $\tree_{t t^\prime}$ is the reduced state of intervention slots $t$ and $t'$ obtained by inserting the Identity instrument ($\cup$) in all remaining slots of the process, 
${\mathrm{T}_t}$ denotes the partial transpose with respect to the slot $t$, and $\|.\|_1$ denotes the trace norm. We first numerically computed the negativity for the uniform process tree, as in Sec.~\ref{sec:multitime}, for Haar random sampling of $U_2$ and $U_1$, averaging over all possible $t,t^\prime$ for each value of $\Delta t$. This was for the reduced Choi state of the process tree for two times, with identity instruments everywhere else, as in the graphic of Fig.~\ref{fig:numerics} (iii). Note that a $\ket{0}$ state preparation is contracted with the very final wire, as by temporal causality constraints (Eq.~\eqref{eq:causality}) this wire cannot be correlated with the others. Remarkably, we found that the negativity tended to decay to zero after at most $\Delta t =10$ time steps, for almost every trial of a random tree. Therefore, having non-zero entanglement in time across long times is apparently not a typical property of a homogeneous, $\mc{W}-$class process tree.

To see whether a process tree can possibly exhibit strong entanglement in time, we next computed the negativity for unitaries parameterized by a Hamiltonian, $U_a = \exp[-i H_{\beta,a}]$, where for $a \in \{1,2 \}$, 
\begin{equation}\label{eq:heisenberg}
    H_{\beta,a} = (1-\beta) \frac{a \pi}{6}  H_{\mathrm{SWAP}} + \beta H_{\mathrm{GUE}}.
\end{equation}
Here, $H_{\mathrm{SWAP}}:= \sigma_x \otimes \sigma_x  + \sigma_y \otimes \sigma_y + \sigma_z \otimes \sigma_z$ generates the SWAP unitary, with $\sigma_a$ the Pauli$-a$ matrix, and $H_{\mathrm{GUE}}$ is sampled from the (chaotic) Gaussian unitary matrix ensemble. Therefore for $\beta=0 $ we have $U_2 = \mathrm{SWAP}^{2/3}$ and $U_1 = \mathrm{SWAP}^{1/3}$, whereas for $\beta=1 $ we recover the random unitary case. We tested between these two extreme cases and found that entanglement in time drops off quickly with $\beta$; heuristically with the amount of randomness in the unitary. On the other hand, for small values of $\beta$, we found that negativity decays polynomially with $\Delta t =  t^\prime - t$. The results are plotted in Fig.~\ref{fig:numerics} (iv). Note that for any $\beta$, the {resultant process tree exhibits} power law decaying temporal correlations (and memory), as in Fig.~\ref{fig:numerics} (ii)-(iii). These results suggest that the scrambling character of $\mc{W}$ (specifically of the unitaries $U_2$ and $U_1$) determines whether the polynomial decaying correlations in a process tree are quantum or classical. 

These numerical results suggest that while process trees generically exhibit power-law correlations, it is not typical that this is genuinely quantum in nature. This is likely due to the cumulative addition of new environment wires in the process tree, with growing time difference $\Delta t$. Each of these wires need to be traced out, by assumption of them corresponding to (inaccessible environmental) degrees of freedom. Looking at Fig.~\ref{fig:circuit}, we see that even for a height $N=3$ process tree, corresponding to $\Delta t \leq 7$, there are seven trace operations. Then if the unitaries $U_1 $ and $U_2$ are sufficiently scrambling, every trace operation potentially introduces classical noise and correlation into the process. However, we have shown in Fig.~\ref{fig:numerics} (iv) that there is the capacity for the process tree to exhibit genuinely quantum temporal correlation across long times. This can be interpreted as a kind of `quantum $1/f$ noise': power-law correlation that cannot be produced from a classical model. 

\section{Process Tree Representations of Spin-Boson Models} \label{sec:grug}
We have so far constructed and analyzed a class of models for capturing the essence of processes with long-range memory. However, the generality of the model has not yet been established, and moreover the connection to physically {reasonable} open dynamics is not clear. Although we have shown that process trees \emph{can} capture nominally interesting physics with these properties, is this actually reflected in reality?
We now turn to this problem by considering a prototypical example of non-Markovian open quantum dynamics in the spin-boson model, and how well it can be described by process trees. 

This particular model has been well-studied in the literature and has many appealing properties that make it a useful class of dynamics to study -- both in the analysis of its physics and in the benchmarking of numerical methods~\cite{Strathearn2018, jorgensen_discrete_2020}. In particular, the order parameter governing system-bath coupling initiates a BKT-class quantum phase transition~\cite{le2008entanglement,LeHurrBook2010}. 
Moreover, as we shall show, it exhibits polynomially decaying correlations in its non-Markovian memory, making this a compelling candidate for the process tree model.

\subsection{Spin-Boson Model}
Quantum impurity systems are a widely studied class of physical problems, with relevant contexts in biology, condensed matter physics, and noise in quantum information processors.
Various subclasses of these dissipative models exist, including that of the spin-boson model, where a two-level system couples to an environment of bosonic modes. 
The corresponding Hamiltonian for our particular setup takes the form
\begin{equation}
\label{eq:SB-ham}
{  H = \Omega \sigma_x + \sum_i \sigma_z(g_ia_i + g_i^\ast a_i^\dagger) + \omega_i a_i^\dagger a_i =: H_S + H_E,}
\end{equation}
where $\Omega$ is the tunneling amplitude between $Z$-eigenstates, $\sigma_x,\sigma_z$ are Pauli spin operators, $g_i$ coupling coefficients, and $a_i^\dagger,a_i$ bosonic creation and annihilation operators of an environment mode with energy $\omega_i$. 
The internal bath correlations are governed by the spectral density of bath frequencies.
A well known spin-boson variant -- the \emph{Ohmic} model -- is the case where the bath has a linear spectral density $J(\omega)\sim \alpha \omega$. The dimensionless parameter $\alpha$ determines the coupling between the system and the environment (or the strength of dissipation). More specifically, the spectral density we work with has an exponential cutoff at some frequency $\omega_c$. That is, $J(\omega) = 2\alpha\omega\exp(-\omega/\omega_c)$.
Here, the BKT phase transition occurs at a critical value of $\alpha$, denoted by $\alpha_c$. The location of this phase transition has been well-studied in the literature; its exact value depends on the chosen cutoff frequency, $\alpha_c = 1 + \mathcal{O}(\Omega/\omega_c)$.
Below $\alpha_c \approx 1$, the system takes a localized phase before discontinuously jumping to a delocalized phase at $\alpha=\alpha_c$~\cite{Bulla_2003}. For a complete discussion of these properties, see Refs.~\cite{le2008entanglement,LeHurrBook2010}.

For the Hamiltonian in Eq.~\eqref{eq:SB-ham}, a process tensor may be constructed representing all of the multi-time correlations of the spin-boson model by, at each step, evolving $H$ for some time $\delta t$ where one half of a fresh Bell pair plays the role of the system at each time. This is the standard formulation of the generalized Choi state~\cite{processtensor}, as detailed in Sec.~\ref{sec:PT}.
To accomplish this, we use the \texttt{OQuPy} software package~\cite{OQuPy} which employs the PT-TEMPO algorithm~\cite{Strathearn2018, fux_efficient_2021} to solve the MPO representation of the $k$-step process tensor for the corresponding spin-boson Hamiltonian given in Eq.~\eqref{eq:SB-ham}.\textsuperscript{\footnote{Alternatively, for larger simulations it may be more efficient to implement a more recent algorithm, such as the `Divide and Conquer' method from Ref.~\cite{cygorek_sublinear_2023}.}} Our parameter choices are a frequency cutoff of {$\omega_c = 10 \text{ ps}^{-1}$} and a step spacing of {$\delta t = 0.1\text{ ps}$}. {Note that since the system Hamiltonian does not factor into the PT-TEMPO computation, we are probing only the influence of the bath. This means that when computing the non-Markovianity, we will be operating in the $\Omega=0$ regime, where the phase transition occurs at $\alpha_c=1$. However, when we shall perform our fits, the process tensors with and without $H_S$ are related up to Trotter error by local rotations. Thus, up to Trotter error, the absence of $H_S$ will not affect the quality of the fits.} We designate this exact (up to numerical truncation) process tensor as $\Upsilon_{k}^{(\text{SB})}$. 
This is the class of process to which we fit our variational process tree models, which we will now detail.

\subsection{Fitting Methods} 
Let $\tilde{\Upsilon}_{N}^{\text{(tree)}}$ be our model for the process using the tree ansatz, with tilde $\tilde{\cdot}$ representing variational models for the extent of this section. In particular, for $k=2^N$ steps, this object encodes $\sum_{i=0}^{N-1} 2^i$ $\mathcal{W}$-bricks, which each implicitly depend on two two-body unitaries. Let us denote this by
\begin{align*}
    \tilde{\Upsilon}_{N}^{\text{(tree)}}&\equiv\tilde{\Upsilon}_{N}^{\text{(tree)}}[\vec{\mathcal{W}}_1,\vec{\mathcal{W}}_2,\cdots,\vec{\mathcal{W}}_N],\\
    \text{where}\quad &\vec{\mathcal{W}}_1 = \{\mathcal{W}_1^{(1)}\},\\
    \vec{\mathcal{W}}_2 &= \{\mathcal{W}_2^{(1)}, \mathcal{W}_2^{(2)}\},\\
    &\vdots \\
    \vec{\mathcal{W}}_N &= \{\mathcal{W}_N^{(1)}, \mathcal{W}_N^{(2)},\dots,\mathcal{W}_N^{(2^{N})}\}.
\end{align*}
Each $\mathcal{W}_s^{(j)}$ corresponds to the $j$th $\mathcal{W}$-block at the $s^{\text{th}}$ level of the tree, as in Sec.~\ref{sec:treeconstruct}. Moreover, each $\mathcal{W}_s^{(j)}$ itself depends on two unitary operations $U_1,U_2\in SU(d_Sd_E)$, parametrized by some $\vec{\theta}$ and $\vec{\phi}$, respectively. Note that the space $E$ of scale $s$ is equal to the $S$ space of the coarser scale $s-1$ implying that $d_S=d_E:=D$. We can hence readily count the number of free parameters in this variational ansatz $\tilde{\Upsilon}_{N}^{\text{(tree)}}(\{\vec{\theta},\vec{\phi}\})$ as $2^{N+1}(d_S\cdot d_E)^2 + \sum_{i=1}^{N-1} 2^{i+1} D^4$.

Suppose now that we have access to a representation of a multi-time process given by process tensor $\Upsilon_{k}^{(\text{target})}$. This process tensor may be in any form (quantum or classical), so long as we have access to inner products with it. 
As an objective function, we take the 2-fidelity between the variational process tree, and the true spin-boson process tensor as target:
\begin{align}
\label{eq:SB-objective}
    &\mathcal{F}_2(\tilde{\Upsilon}_{N}^{\text{(tree)}},\Upsilon_{k}^{(\text{target})})  \\ \nonumber \mbox{with} \quad &\mathcal{F}_2(\rho,\sigma) := \frac{\tr(\rho\sigma)}{\max[\tr(\rho^2),\tr(\sigma^2)]}.
\end{align}
This measure is not exactly the Uhlmann fidelity, but has the desirable properties that $\mathcal{F}_2(\rho,\sigma)\in [0,1]$, and $\mathcal{F}_2(\rho,\sigma) = 1 \iff \rho=\sigma$. Moreover, it is symmetric, unitarily invariant, and reduces to the Uhlmann fidelity in the case where either $\rho$ or $\sigma$ is pure~\cite{Liang_2019}. These features make this an ideal objective function for fitting mixed quantum states. Importantly, the 2-norm is far easier to deal with in the case of tensor networks than the 1-norm, rendering the subsequent computations feasible. We remark that evaluation of Eq.~\eqref{eq:SB-objective} for generic $\Upsilon_{k}^{(\text{target})}$ is inefficient. In practice, however, we find that we can compute it for single-qubit processes up to 128 steps on a personal laptop. 

Although the process tree construction clearly embeds multi-time physics, it is only efficient in time steps $k$ for computing $\ell$-body expectation values (with $\ell$ fixed), where every other step is projected onto the identity supermap. In order to fit a process tree with an \emph{efficiently} evaluated objective function, then, one could always fit the tree to $\ell$-body marginals of the target process. This may be pertinent (for example, with the spin-boson model) in instances where important physics is captured by low-weight correlations in the dynamics. Here, we are interested in examining the overall expressiveness of the ansatz and hence do not employ this approach.

Now that we have an objective function and a parametrized tensor network, it is straightforward to fit a process tree to given dynamics. We use the L-BFGS-B optimizer~\cite{lbfgs-alg} to maximize Eq.~\eqref{eq:SB-objective} with a \texttt{JAX} backend~\cite{jax2018github} to compute the gradients via automatic differentiation, as well as the Python library \texttt{quimb}~\cite{quimb} to handle tensor network semantics. Rather than embedding the unitarity condition into $U_1$ and $U_2$, the operators are permitted to be general and then unitized at each step. The optimization procedure is highly non-linear. Hence, each process fit is run many times from different random seeds until appropriately converged.
Typically, however, we find that this optimal value is quickly found. Only occasionally does the optimizer get stuck in a local minima, but with the addition of stochastic methods (such as a change from L-BFGS-B to Adam~\cite{kingma2014adam}) the methods are highly reliable. 

We finally remark that we have an extra freedom to explore in the parametrisation of $\Upsilon_{N}^{(\text{tree})}$ -- namely, the invariance of $\mathcal{W}$ across different time steps and time scales. In particular, we can further compress the description of the dynamics by enforcing a type of homogeneity of the tree, where $\mathcal{W}$-bricks in different locations made to be the same as one another. There are two important cases we consider in that respect, a time-homogeneous tree and a scale-time-homogeneous tree. The former supposes that dynamics are invariant on translation across time (same $\mathcal{W}$ block within a level but different between levels) and the latter is invariant across both time and scale (same $\mathcal{W}$ block in each position of the tree). That is, we set $\mc{W}^{(j)}_s \equiv \mc{W}_s$ in the former case, and $\mc{W}^{(j)}_s \equiv \mc{W}$ in the latter.
We denote these respectively by $\Upsilon_{N}^{(\text{th tree})}$ and $\Upsilon_{N}^{(\text{sth tree})}$, such that
\begin{equation}
    \begin{split}
        \Upsilon_{N}^{(\text{th tree})} &\equiv \Upsilon_{N}^{(\text{th tree})}[\mathcal{W}_1,\mathcal{W}_2,\cdots,\mathcal{W}_n]\\
        \Upsilon_{N}^{(\text{sth tree})} &\equiv \Upsilon_{N}^{(\text{sth tree})}[\mathcal{W}].
    \end{split}
\end{equation}
This constitutes all the required machinery to fit process tree models to target (here, spin-boson) multi-time processes.

\subsection{Results}
There are two properties of the tree we wish to examine. First, we are interested in determining whether the $\mathcal{W}$-brick construction is expressive enough to describe real physical systems. Next, we wish to explore the extent to which homogeneity of the tree can play a role both in compressing the model and in generalizability: constructing larger trees from smaller constituents to describe future time dynamics. We will first demonstrate the long-range nature of the spin-boson model through an analysis of its non-Markovian memory; then investigate the performance of process tree fits for a range of coupling strengths; and finally determine how smaller process tree fits generalize to larger ones.

\begin{figure}[t!]
    \centering
    \includegraphics[width=0.95\linewidth]{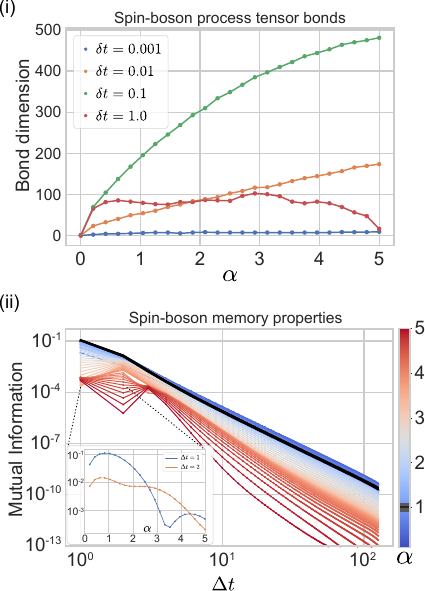}
    \caption{{Properties of the PT-TEMPOs for the spin-boson model within which we work. (i) Comparison of the maximum bond dimension of each MPO as a function of coupling strength $\alpha$, across different choices of time step $\delta t$. We see that when this is chosen too small, there is little interesting memory structure; too large and the system starts to thermalize, also removing memory. The time step $\delta t =0.1$ is therefore chosen for all further plots. (ii) Memory structure of the spin-boson model. We determine the scaling of non-Markovianity as a function of temporal distance $\Delta t$ for different coupling strengths $\alpha$ (with $\delta t =0.1$). From the linear plot, it is clearly seen that this process has polynomially decaying correlations. Inset: the short-time correlations as a function of $\alpha$; these are maximized close to the phase transition.}}
    \label{fig:sb-properties}
\end{figure}

\emph{Spin-boson memory structure.--} We start by providing some motivation for this analysis. The spin-boson model is well-known to feature long-range interactions in time, usually explored via its mapping to an Ising spin chain with long-range interactions~\cite{le2008entanglement}. We show here first how this corresponds to long-range temporal correlations as we have motivated it so far throughout this work. {First, it is instructive to look at the bond dimension scaling of the PT-MPOs for a 128 time-step spin-boson process tensor, given in Figure~\ref{fig:sb-numerics}(i). With a relative truncation threshold of $10^{-7}$, we can see how the choice of evolution times $\delta t$ influences the resulting maximum bond dimensions. This informs our choice of $\delta t = 0.1 \text{ ps}$: this timescale seems to capture the most interesting memory dynamics.}
Starting with the numerically computed spin-boson process tensor across $128$ time-steps $\Upsilon_{128}^{\text(SB)}$, we sweep $\eta(\Upsilon_{128}^{\text(SB)}; t, t+\Delta t)$ (Eq.~\eqref{eq:NM1}) to find the average quantum mutual information for different values of $\Delta t$. The results of this are shown in Fig.~\ref{fig:sb-properties}(ii) across a wide range of coupling strengths from 0.1 to 5 in steps of 0.1. From this, we can see several interesting features. First, it is clear that temporal correlations are decaying polynomially, corroborating the notion of a spin-boson model as possessing long-range memory. This is true for almost all values of $\alpha$.
However, it appears to be a coupling-dependent observation. Although correlations decay polynomially for most values of $\alpha$, this is not true for very low or very high values. At $\alpha=10^{-4}$ (not shown), for example, the non-Markovianity curves show exponential decay in time. Similarly, one can see that with increasing $\alpha$, the curves in Fig.~\ref{fig:sb-numerics}(ii) begin to lose their linear character on the log-log plot. Respectively, this could be understood as an extremely weak environment incapable of mediating this slow memory, and a thermalized environment where information effectively only dissipates outwards. {The non-monotonic behavior at very early times likely results from the effect of the causal break operations; $\rho_a$ and $\rho_b$ in Eq.~\eqref{eq:NM1}. We expect these independent measurement/preparations to affect correlations most at earlier times, before their influence has a chance to effectively dissipate through the environment. Further, these will be more pronounced for stronger $SE$ couplings. However, this is only speculative: it is a non-equilibrium system and may have other pathologies in the dynamics.}

\begin{figure*}[t]
    \centering
    \includegraphics[width=\linewidth]{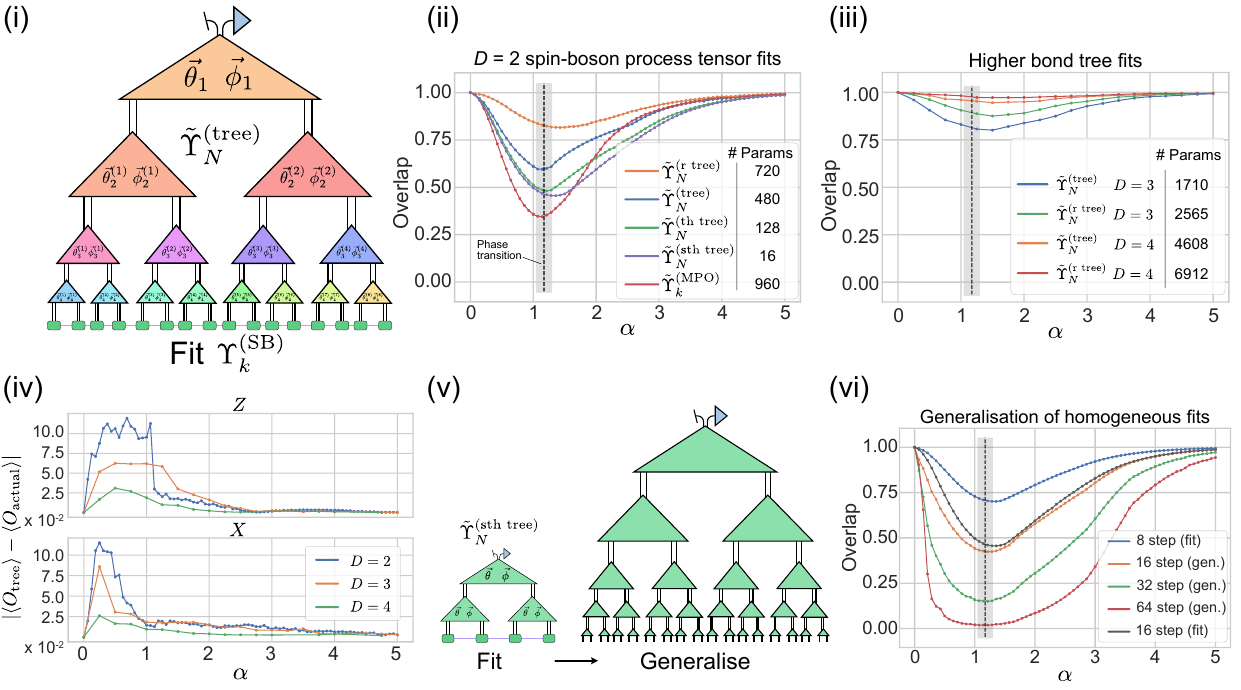}
    \caption{{Procedure and results fitting process tree ans\"atze to physical models. (i) Illustrative diagram of the fitting procedure. We parametrize each $\mathcal{W}$ brick of the process tree in terms of its unitary constituents. The tree is then optimized to maximize $\mathcal{F}_2$ with respect to some target process (in this case the spin-boson process tensor).  (ii) Results fitting different $D=2$ process ans\"atze to the 16-step spin-boson process tensor. We report the best overlap of various trees as well as a fixed bond dimension MPO with respect to the stated model across a range of coupling strengths $\alpha$. Not only does the tree well-capture the intended physics of the process, we clearly see the effects of the spin-boson BKT-transition as correlations increase. 
    (iii) Results fitting $D=3$ and $D=4$ process trees to the spin-boson process tensor, both with the strict $\mathcal{W}$-brick parametrization, and with the relaxed $\mathcal{W}$.
    (iv) Computing accuracy of observables of the spin from the $\mathcal{W}-$tree fits in (ii), averaged over the 16 time steps. We start the spin in $|0\rangle$, apply a $\sigma_x$ propagator, and then measure $\langle Z \rangle $. We also start the spin in a $|+\rangle$ state, let it idly evolve, and then determine $\langle Z\rangle$.
    (v) Illustrative diagram of generalizing scale-time-homogeneous trees to larger processes. After fitting the $\mathcal{W}$ brick of a smaller process, one can then construct a larger generalization out of that optimized brick. (vi) Results showing the generalization of smaller fits to larger ones. A scale-time-homogenous process is fit to an 8-step spin-boson process tensor and its optimized $\mathcal{W}$ brick used to construct 16, 32, and 64 step processes. }}
    \label{fig:sb-numerics}
\end{figure*}

It is not clear the extent to which the polynomial decay at other $\alpha$ are finite-sized effects. Away from the phase-transition, the curves start to deviate from their linear character for high values of $\Delta t$. For the purposes of this work, it suffices to observe the polynomial decay to demonstrate both the efficacy and the suitability of a tree ansatz, but in a future work it would be interesting to see whether this behavior holds at much longer times, and whether it is truly only at the phase transition where this correlation length is effectively infinite. Although Fig.~\ref{fig:sb-numerics}(ii) does not represent all of the properties of the process, it is interesting to note that the {short-range} non-Markovian memory is maximized at (and near) the phase transition.
The nature of non-Markovianity or multi-time correlations has not been studied in this context to our knowledge, but serves as a platform to further understand the interesting physics of this paradigmatic model. 

\emph{Process tree fits.--} Having established the memory scaling behavior of spin-boson models, we can investigate the extent to which simple process trees represent the physics. The setup to the fitting problem (as discussed in the previous subsection) for a generic tree is depicted in Fig.~\ref{fig:sb-numerics}(i). For a 16 step process, we generate $\Upsilon_{16}^{\text{(SB)}}$ for a range of values of $\alpha$. For the sweep of $\alpha$, we then find the best-fit for five different models: a $\mathcal{W}$-brick process tree $\tilde{\Upsilon}_4^{(\text{tree})}$, a time-homogeneous tree $\tilde{\Upsilon}_4^{(\text{th tree})}$, a scale-time-homogeneous tree $\tilde{\Upsilon}_4^{(\text{sth tree})}$, and an MPO $\tilde{\Upsilon}_{16}^{(\text{MPO})}$. {We also consider a process tree with a `relaxed' superprocess condition, which we denote by $\tilde{\Upsilon}_N^{(\text{r tree})}$. It is constructed out of $\mathcal{W}$-bricks with unitaries that are not constrained to multiply out to the identity. That is, $\Lambda_1 = U_1$, $\Lambda_2 = U_2$, and $\Lambda_3 = U_3$ as per Eq.~\eqref{eq:y_brick} -- they are not scale-consistent. We include this last tree to test the restrictiveness of the $\mathcal{W}$ brick in comparison to a proxy for $\mathcal{Y}$; recalling that $\mathcal{W}$ allows for efficient computation of two-point correlations, but is a restriction which is not necessarily the best choice.}
{We trial these various models for an environment size of $D=2$, 3, and 4 (which are respectively the square roots of their internal bond dimensions. The system size (bottom layer) is kept constant at $d_S=2$.}

The results of the various model fits can be seen in Figs.~\ref{fig:sb-numerics}(ii), (iii), and (vi). {Additionally, the tree fits for the $\mathcal{W}$-brick process tree $\tilde{\Upsilon}_4^{(\text{tree})}$ from (ii) and (iii) are used to compute the state of the spin as a function of time in two scenarios: in the first, the system is prepared in the $\ket{0}$ state, and a $\sigma_x$ propagator $\exp(-i\delta_t \sigma_x / 2)$ applied at each time, and $\langle Z\rangle$ determined. In the second, the system is prepared in the $\ket{+}$ state, left with $H_S=0$, and then $\langle X\rangle$ computed. These values are then each compared with the PT-TEMPO computation. Naturally, the error is larger in the delocalized phase before the spin effectively freezes. This provides an additional metric to the $\mc{F}_2$ overlap. Note however, that since the trees are not optimized for these two-point correlations, we would not necessarily expect this to be the best tree representations to compute these quantities.} 
We see several remarkable results in the remainder of plots. The first is to note that the generic process tree serves as a very good fit for the spin-boson model{, even for small internal bonds ($D=2$) but especially as we increase the bond slightly}. 
This demonstrates that our approach to efficiently characterizing long-range temporal correlations in not only of academic interest, it practically captures the physics of well-established non-Markovian open quantum systems.
Moreover, it is interesting to note that these fits witness the phase transition with $\alpha$. The minimum value of the $D=2$ curves coincide with the value of $\alpha_c = 1 + \mathcal{O}(\Omega/\omega_c)$ {(since the exact location of the phase transition is dependent on the system propagator, we have indicated some uncertainty here.)}

Although the reduction in fit quality coincides with an intuition of criticality at the phase transition, it is not entirely clear what singular property causes this breakdown. As observed in Fig.~\ref{fig:sb-properties}(ii), the short-range memory is maximized at the phase transition and remains polynomially decaying. However, the non-Markovianity described here is perhaps too coarse a measure to fully describe the fitting results. For instance, taking $\alpha=0.5$ gives nearly the same scaling of the memory, but the tree is a substantially better fit. It is likely that finer measures of the process complexity (such as multi-time correlations or spectral properties) are needed to sufficiently answer this question.
Nevertheless, we are not aware of this perspective on the spin-boson model phase transition elsewhere in the literature, and believe it warrants further study. {It would also be interesting to investigate the exact role played by the use of $\mathcal{W}$ superprocesses in this context. The relaxed tree fits clearly perform much better, but do not have substantially more free parameters. This may hint at a process wherein scale consistency does not immediately hold.}

The next feature worthy of attention is to compare between the fits of $\tilde{\Upsilon}_4^{(\text{sth tree})}$ and $\tilde{\Upsilon}_4^{(\text{th tree})}$. The two models are very closely matching for $\alpha < \alpha_c$ and only start to deviate once the system transitions from the delocalized to localized phases, despite an order of magnitude difference in the number of free parameters.
We might understand this as also representing a change in the memory structure to something less homogeneous in scale. This is suggestive that complex dynamical systems may exhibit a temporal change-of-scale invariance, in close analogy to critical many-body states and the tensor network representations thereof, such as MERA~\cite{CriticalMERA}.

Lastly, and perhaps most compelling to the case for the process tree in this instance,
there is a large gap between the optimized MPO model and the different process tree fits, particularly near the phase transition. 
In this small instance, a fixed bond dimension MPO is still expressive enough to {significantly} overlap with the spin-boson process tensor. 
However, despite possessing substantially more (2-60$\times$) free parameters than the different process tree models, the MPO is not appropriately structured to match the polynomial decay of the temporal correlations in the spin-boson model. We present this as evidence that not only are process trees better suited to represent complex dynamics, it is also {a more} efficient representation.

\emph{Generalizing scale-time-homogeneous models.--} Our last set of results constitutes an investigation into the utility of these models in predicting future dynamics. A generic tree fit might be highly tailored to that specific process and set of steps, rendering future predictions inaccessible. However, if the process exhibits homogeneous features, then it should be possible to assemble a larger future representation from a smaller process. This is depicted graphically in Fig.~\ref{fig:sb-numerics}(ii). In particular, here we fit a variational tree $\tilde{\Upsilon}_3^{\text{(sth tree)}}$ to an 8-step spin-boson process tensor. Recall that this process tree is fully characterized by a single $\mathcal{W}$ brick at all times and scales. It is therefore straightforward to construct a larger process $\tilde{\Upsilon}_N^{\text{(sth tree)}}$ from this optimized $\mathcal{W}$. 

Once more, we sweep a range of coupling values $\alpha$. The optimal $\mathcal{W}$ fitting to an 8-step spin-boson process tensor is then used in the best case to construct a 16, 32, and 64 step process, and their overlaps computed with the true spin-boson processes for these higher steps. Our results are shown in Fig.~\ref{fig:sb-numerics}(v), alongside the 16-step fit $\tilde{\Upsilon}_3^{\text{(sth tree)}}$ from Fig. ~\ref{fig:sb-numerics}(iv) as a point of comparison. 
Although the quality of the curves decays with the increased number of steps, there is still a surprising amount of overlap in the larger instances. Quantum states at this scale are extremely unlikely to have any coincidental overlap. Clearly, there is some essential physics that can be captured and generalized from small processes without explicitly solving for the Hamiltonian.
Although only preliminary, this points at the ability for the tree to not only be efficient in its representation, but to be predictive of future dynamics of a system.

\section{Process Trees From Microscopic Physics} \label{sec:plaquette}
An interesting remaining question is that of how to construct a process tree from the underlying dynamics. Although we have so far studied representability and the in-principle expressivity of the process tree ansatz, if one were interested in explicit simulation of a non-Markovian open quantum system, it is often desirable to take an additional step and start from the underlying dynamics. When supplied with a system-environment Hamiltonian, the prototypical method to obtain a process tensor is to employ the Feynmann-Vernon influence functional (IF)~\cite{Feynman-1963}, which encodes environment correlations in an effective trotterisation of the full dynamics. The resulting functional generally grows exponentially with time, but can be exactly represented by a 2D tensor network, whose contraction can typically be hard to compute. In this section, we pose a resolution to this problem as an application of the tensor renormalization group (TRG) scheme~\cite{Vidal2007,Levin_2007,PhysRevB.78.205116,Evenbly2015,Evenbly_2017}. 
This can be used to iteratively construct a process tree from the IF. 
We leave, however, the problem of explicit numerical analysis and benchmarking against the state-of-the-art numerical packages -- such as ACE~\cite{Cygorek2022,cygorek_sublinear_2023} or TEMPO~\cite{Strathearn2018,Pollock2019,fux_efficient_2021} -- to future work.

\begin{figure*}[t!]
    \centering
    \includegraphics[width=\linewidth]{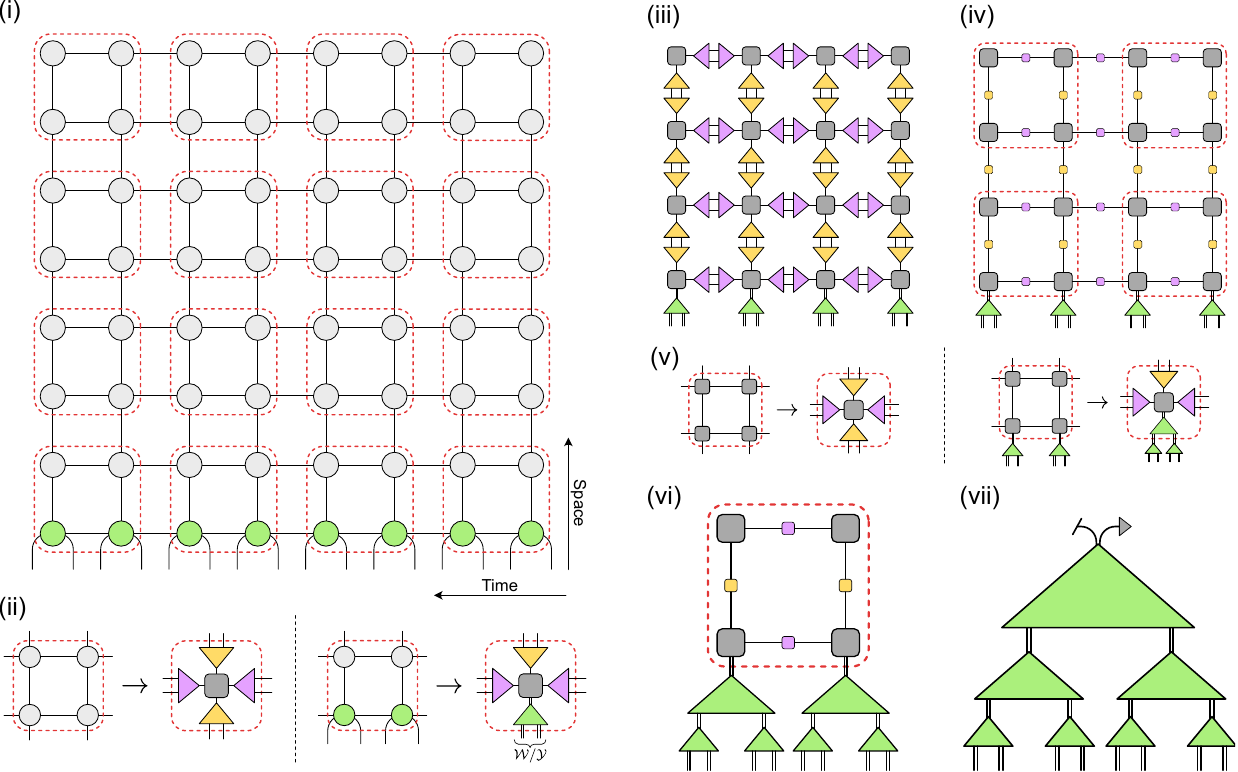}
    \caption{{Coarse-graining scheme to construct a process tree from an influence matrix. (i) System-environment dynamics are written exactly as a 2D tensor network using the influence matrix formalism. This tensor network is in general inefficient to contract. We propose a series of plaquette moves whereby groupings of the influence matrix are coarse-grained. At the bottom row (in green) is the `system' component of the influence matrix. (ii) Each group of four sites is replaced with a close tensor network that possesses a structure amenable to renormalization. The system components that later comprise the tree can be put into $\mathcal{W}$ or $\mathcal{Y}$ form. (iii) The IM after one round of plaquette moves reduces to (iv) wherein corresponding horizontal (purple) and vertical (yellow) triangles multiply out and reduce the internal dimension. These matrices may then be absorbed into the (grey) box tensors. (v) The next round of transformations repeats the plaquette moves on the newly renormalized lattice, resulting in the next stage at (vi), and the final process tree at (vii).}}
    \label{fig:influence-trg}
\end{figure*}

Suppose we have a spin system coupled to an external bath via the time-independent Hamiltonian $H=H_0 + H_E$. The time evolution operator $\mathcal{U}_{\delta t}$ for this open system can be approximated to first order as $\mathcal{U}_{\delta t} = V_{\delta t/2}W_{\delta t} V_{\delta t/2}$, where $V_t = \exp(-iH_0t)$ is the free evolution of the system, and $W_t = \exp(-i H_E t)$ is both the free bath evolution and any interactions. Then, for a finite time $t$ discretized into $k$ steps, the time evolution operator can be written as $\mathcal{U}_t = \prod_{i=1}^k V_{\delta t/2}W_{\delta t} V_{\delta t/2}$, up to an error $\mc{O}(\delta t ^3)$~\cite{Childs_2021}. If we take the initial state over $\mc{H}_{S} \otimes \mc{H}_{E}$ to be a product state $\rho_S\otimes \sigma_E$, then it was shown in Ref.~\cite{Pollock2019} that the process tensor can be arbitrarily approximated via 
\begin{equation}
    \begin{split}
        \Upsilon_{k}&\approx (V_{\delta t /2}\otimes V_{\delta t/2}^\ast)^{\otimes k} [\mathcal{F}_{k}]\otimes \rho_0,\quad\text{where,}\\
        \mathcal{F}_{k} &= \sum_{\vec{s},\vec{r}}\text{Tr}_E\left[W_{\delta_t}^{(s_k,r_k)}\cdots W_{\delta_t}^{(s_1,r_1)}[\sigma_E]\right]\times \\
        & |s_k s_k\cdots s_1 s_1\rangle\!\langle r_k r_k\cdots r_1 r_1|.
    \end{split}
\end{equation}
The matrix $\mathcal{F}_{k}$ is known in the literature as the \emph{influence matrix} (IM) -- the matrix form of the IF. Naturally, in its dense form this is exponentially large, but this complexity can be simplified down in both space and time, and it can represented as a tensor network~\cite{Strathearn2018,Lerose2021,Ye_2021,Cygorek2022}. For an environment that factorizes across a finite number of modes $N$, the Liouville superoperator representing the time increment $W_{\delta_t}$ can be written as a matrix product operator with some exact bond dimension that depends on the locality of the interaction. \par 

After $k$ steps of the process are composed together, then pre-contraction the process tensor will be representable as a 2D tensor network, as depicted in Fig.~\ref{fig:influence-trg}a. Here, we delineate the system in green (bottom row), which has free indices at each time as inputs and outputs. The causal structure of this tensor network can be exploited to give it a canonical form~\cite{Pollock2019}, circumventing computational hardness issues in the contraction of a lattice~\cite{haferkamp2020contracting}. Depending on the exact contraction scheme employed (see Ref.~\cite{Ye_2021} for an in-depth analysis), the computational cost is $\mathcal{O}(k\chi_{\rm t}^2 \chi_{\rm s}^2 d_S^2$), where $\chi_t$ and $\chi_{\rm s}$ are the respective temporal and spatial bond dimensions in the network. In the worst case, $\chi_{\rm s}$ can grow like $d_S^{k}$. In many physical instances, however, this is not the case, and one can include bond truncation in the final consideration without sacrificing accuracy. As a concrete example, it was shown analytically in Ref.~\cite{vilkoviskiy2023bound} that the bond dimension of the spin-boson model IF scales polynomially both in the simulation time and the error on the physical observables. Further, the influence matrix of a range of integrable lattice models turn out to also exhibit sub-exponential bond dimension scaling; low `temporal entanglement'~\cite{Lerose2021,Lerose2021prb,Thoenniss2023}.\par

We now show how -- using the TRG approach to renormalize the IM -- we can iteratively construct the process tree via a series of plaquette moves. A similar renormalization was performed in Ref.~\cite{Evenbly_2017} with respect to a lattice of trotterized imaginary time evolution. Our contribution here is firstly the insight that the same type of lattice can be used to represent an IM, and that we can keep the bottom layer free to represent the process tensor. Crucially, the process of renormalizing in this fashion places the process tensor exactly into tree form. At each step of coarse-graining, we simply need to encode our chosen $\mathcal{W}$ or $\mathcal{Y}$-brick structure into the plaquette elements that contain the process tensor. Unlike the direct tree fits, this process does not scale at all with the number of steps, and so can be performed efficiently. Note, however, that it does scale with the size of the environment. 

Consider first a base unit cell of the lattice, such as the dotted-line (red) boxes depicted in Figs.~\ref{fig:influence-trg}(i) and (ii). 
Since each column in this tensor network represents a propagator, it can be placed into an appropriate canonical form~\cite{Ye_2021}. The two in-going horizontal indices then propagate a sub-part of the system-environment state, and the top and bottom indices are coupling terms between different environment modes at different times. In principle, then, from right to left, we then have a trace-preserving map, and from bottom to top we have a superprocess. We leave these possible identities in place, but unless we wish to cut the internal bonds of the IM, it is unnecessary to preserve this structure. Instead, suppose that we can find a plaquette coarse-graining transformation that replaces these four IM sites with a structured network -- shown in Fig.~\ref{fig:influence-trg}(ii) -- incurring minimal error. The triangles here 
are a mapping from the $i$th and $j$th environment modes to some artificial space of size $\chi$: $\mathcal{B}(\mathcal{H}_{E_i})\otimes\mathcal{B}(\mathcal{H}_{E_j})\to \mathcal{B}(\mathcal{H}_{\chi_l})$. This represents some dimensionality reduction of the problem, and in practice each $\chi_l$ can be chosen based on the desired numerical precision. 


It is not obvious how to analytically construct this decomposition of the IM unit cell. As such, we defer to numerical methods that would greedily find a best fit. Note that this is not unusual in the tensor network literature: for example, the MERA has no known analytic form for computing its disentanglers from an arbitrary underlying state~\cite{evenbly_algorithms_2009}. In particular, each element of the plaquette can be individually parametrized. For the bulk of the IM, we can then optimize for the distance with respect to the unit cell:
\begin{equation}
\includegraphics[width=0.65\linewidth, valign=c]{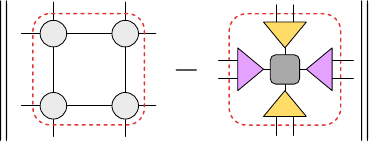}. \label{eq:plaquette-opt}
\end{equation}
We also need to take care of the bottom level, where the system has open indices for inputs and outputs to the process. This can be done similarly, except that the bottom tensor must now be parametrized as a superprocess -- as per the structure of superprocesses we have so far introduced. Imposing a variational form on this map additionally gives us some flexibility in the final tree structure. For instance, if we wished our tree to be composed of either $\mathcal{W}$ or $\mathcal{Y}$-type bricks, respectively, then we could simply encode this parametrization in the superprocess constituents with each optimization:
\begin{equation}
\includegraphics[width=0.65\linewidth, valign=c]{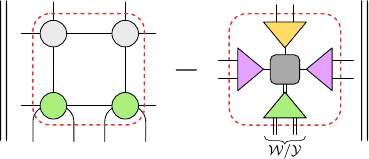}. \label{eq:super-plaquette-opt}
\end{equation}
The particular details of this optimization are not pertinent. We can, of course, adopt a similar strategy as in Section~\ref{sec:grug}. If the IM is translationally invariant in both time and space, then the minimization procedure needs only to be carried out once. Otherwise, we will need $\mathcal{O}(k\cdot N)$ optimizations, which can be carried out in parallel. Finally, we remark that the boundary cases follow similar logic, but will only have three out of the four isometries, reflecting the reduction to six outer indices instead of eight.\par 

With a best-approximation to each unit found, the IM can be recast as the lattice shown in Fig.~\ref{fig:influence-trg}(iii). We note then that the inner products between isometries and superprocesses from adjacent cells reduce into linear maps on the chosen internal bond dimension, depicted in Fig.~\ref{fig:influence-trg}(iv). Both the small boxes on the horizontal (pink) and vertical (yellow) bonds can each be absorbed into new larger box (dark-gray) tensors to define a new lattice. This is apart from the bottom row of coarse-graining superprocesses which have open indices representing input and output legs on $\mc{H}_S$. These remain free, and will iteratively construct the process tree.
This now defines a lattice in which the time and space sites have all been reduced by a factor of four. We can simply repeat this step by coarse-graining the lattice once more, replacing each plaquette as in Fig.~\ref{fig:influence-trg}(v), and end up with the coarse-grained lattice shown in Fig.~\ref{fig:influence-trg}(vi).
We continue this until the IM has been reduced to a single point. The final optimization must now also include the initial environment projection at the very top level. The resulting concatenation of superprocesses mapping the system is now exactly a process tree. This is depicted in Fig.~\ref{fig:influence-trg}(vii).

The purpose of this section is to sketch out an in-principle method by which process trees can be constructed from the microscopic physics which define them. This leaves, naturally, a host of open questions for future work -- including benchmarking, stability, and the determination of best practices for selecting plaquette moves. 
Nevertheless, we have demonstrated that the construction of process trees is well within reach of commonly-studied approaches, both in the context tensor renormalization and the simulation of open quantum systems. For the former, TRG algorithms comprise an extensive body of literature~\cite{Levin_2007,Vidal2007,PhysRevB.78.205116,Evenbly2015,Evenbly_2017,Yang2017}; for the latter, IM methods are key in essentially all state-of-the-art simulation methods for non-Markovian open quantum systems~\cite{Cygorek2022, cygorek_sublinear_2023, cygorek-ace-2024, Strathearn2018,Pollock2019, fux_efficient_2021}. It will be imperative future work to find and characterize open dynamics whose structure is amenable to a tree construction (such as the spin-boson model, which we have demonstrated).

\section{Discussion and Conclusions}\label{sec:discussion}
In this work, we have introduced and {systematically} studied a new class of quantum non-Markovian process tensors (Sec.~\ref{sec:treeconstruct}), which we showed to structurally exhibit power-law correlations and memory (Sec.~\ref{sec:multitime}). We refer to them as \textit{process trees}, and they are naturally suited to represent complex quantum processes. Namely, the process tree can accommodate genuinely quantum long-range correlations (Sec.~\ref{sec:NM-cor}). Moreover, process trees include a notion of temporal scale and are constructed such that observables at different scales and across arbitrarily many times at these scales are efficiently computable. Finally, we have shown that this model can accurately fit the multitime process tensor of the paradigmatic spin-boson model (Sec.~\ref{sec:grug}), {and detailed a method for systematically constructing the ansatz from underlying dynamics (Sec.~\ref{sec:plaquette}).} 

The process tree, including its applications to the spin-boson model, forms a proof of principle model to efficiently represent and simulate complex physical non-Markovian quantum processes. This paper thus opens up many avenues of future research in simulating large-scale complex quantum processes, new tensor network ans\"atze, understanding phases of quantum dynamics, and much more. We discuss each of these avenues in some detail below.

{
Representations of full, many-body quantum states can quickly become expensive using tensor networks. Alternatively, when one is interested in only few-point correlations, isometric tree tensor networks offer efficient encoding methods for long-range systems. In the time domain, we have shown that a similar principle applies. Indeed, existing numerical methods that rely on reconstructing the full process tensor (and hence implicitly encoding many-point correlation functions) similarly tend to have a cost which scales fast. Extending tree tensor network principles to the time domain, the process tree class offers fertile ground for innovations regarding efficient numerical simulation. It remains a pressing task to implement the spatiotemporal TRG method described in Sec.~\ref{sec:plaquette}, or a similar method, for realistic physical models.} Such an algorithm would be predictive, in contrast to the fitting method of Sec.~\ref{sec:grug} which {while informative}, is generally inefficient and relies on an existing process tensor of {the dynamical} system. We stress that this method is underlying-model agnostic, and should apply equally to influence matrices derived from continuous environments~\cite{Strathearn2018,Cygorek2022} or from many-body lattice models~\cite{Banuls2009,Hastings2015,Lerose2021,Lerose2021prb}. {Moreover, the TRG method systematically incorporates temporal environmental scales, such that one could access different levels of coarse observables from the single tensor-network description of the process (tree).} This would be relevant to systems with different physics emergent at different scales, {additionally helping to identify temporal phases of quantum processes.} 

{The model introduced in this work was largely} motivated by the physical models that exhibit long-range temporal correlations, and we adapted structures from tree tensor network theory from many-body physics to {identify} this generic dynamical class. However, MERA tensor networks also exhibit long-range (critical) correlations and have the advantage of not needing averaging to achieve exact polynomially decaying correlations. MERA networks differ from tree tensor networks through alternating layers of two-to-two isometries, called `disentanglers'~\cite{CriticalMERA}. We investigated the possibility of extending process trees to include two-to-two time-slot bricks, such as the $\mc{D}$ and $\mc{F}$ superprocesses in Eqs.~\eqref{eq:D} and \eqref{eq:D1}, but were unable to find an ansatz which exhibited the same desirable properties as the spatial MERA. Part of the problem is the inherent causal ordering of indices in a process, such that one cannot possibly find a nontrivial two-to-two time superprocess that is isometric, i.e. which perfectly preserves information going between scales. It would be interesting to investigate this question further, as it may relate to more formal aspects of change-of-scale (renormalization group) transformations in the temporal setting, including identifying phases and universal features of quantum dynamics. {Recent advances in temporal scale simulation~\cite{PhysRevX.14.021007} and in lattice RG techniques~\cite{PhysRevX.14.021008}, may help in this research program.} 

Interpreting the fine-graining as depicted in Fig.~\ref{fig:circuit}, we can see that going to `finer' time scales (lower down the tree) involves introducing more dynamics (tensor boxes). Physically, this represents the fact that there exists time regimes in dynamics where certain Hamiltonian terms are relevant, and when they become irrelevant. For example, high-frequency modes of a Hamiltonian may be relevant only at a fine scale, but average-out at a coarse scale (such as via the rotating wave approximation~\cite{burgarth2023taming}). One could also connect this with master equation descriptions of open quantum system dynamics, which can be derived exactly from the process tensor~\cite{Pollock2018tomographically}. Looking at Fig.~\ref{fig:causalcone1}(i), we can see that with the scale consistency condition, $\mc{W}-$class processes have a restricted amount of information transferred from the past. That is, adding a single scale layer adds a single $\mc{W}-$brick tensor, but describes a process tensor with twice as many time steps. This fits with the idea of approximate Master equations from a restricted/efficient memory kernel~\cite{Pollock2019,pollockscipost}, and may be related to notions of the complexity of open quantum systems~\cite{isobel2022,guo2022memory}.

Finally, we remark that the process tree can be prepared using only two-body gates in a laboratory setting. Therefore, in principle, process trees can be engineered, and the various structural properties that we have described here could also be verified experimentally. Our results hence pave a way forward for engineering and simulation of complex quantum processes based on the ansatz introduced in this work.

\begin{acknowledgments}
We thank Isobel A. Aloisio for early work and useful conversations. We thank Erik Gauger for his insights into the spin-boson model. KM thanks Amelia Liu, Nathan McMahon, and Felix A. Pollock for early conversations leading to initial ideas. ND thanks Georgios Styliaris for useful technical discussions on tensor networks. ND is supported by an Australian Government Research Training Program Scholarship and the Monash Graduate Excellence Scholarship. KM is grateful for the support of the Australian Research Council via Future Fellowship FT160100073, Discovery Projects DP210100597 and DP220101793.
\end{acknowledgments}


%

\appendix

\section{Tensor Networks} \label{ap:tn}
A \textit{tensor} $T_{ijk...}$ is a multi-dimensional array of real or complex numbers. The tensor components are located within the tensor via a set of indices $i,j,k,...$. 
The number of values an index runs over is the \textit{dimension} of that index. For instance, if $i=1,2,...,d$ then index $i$ has dimension equal to $d$. We denote the index dimension as $|i| = d$.
The size of a tensor -- the number of components it has -- is equal to the product of the dimensions of all its indices.

Common examples of tensors are vectors and matrices. A vector (ket) $|\psi\rangle = \sum_i \psi_i|i\rangle$ inside a vector/Hilbert space $\mathcal{H}$ is described by tensor $\psi_i$ with one index, where $|i\rangle$ denotes a basis in $\mathcal{H}$, and its dual (bra) $\langle\psi| = \sum_i \psi^*_i\langle i|$. (${}^*$ denotes complex conjugation.) Similarly, a matrix $U$ acting on $\mathcal{H}$ is described by a two-index tensor $U_{ij}$ as $U = \sum_{ij} U_{ij} |i\rangle\langle j|$. Higher-order tensors are arrays with more than two indices. For instance, a two-body gate in a quantum circuit -- for example, a CNOT gate -- is described by a four-index tensor $V_{ijkl}$ with two bra indices $i,j$, and two ket indices $k,l$. In fact, CNOT is a \textit{multi-index matrix}, a particular type of higher-order tensor obtained by reshaping a matrix.

We graphically represent tensors as shapes (such as a circle, triangle, or square) with wires emanating from them. See Fig.~\ref{fig:graphical}. The wires correspond to the indices of the tensors. In this paper, most diagrams in this paper are time-ordered, which results in an implicit orientation of ket and bra indices/wires. We follow the convention that time flows from the left to the right, which means that indices emanating to the left (right) of a shape are bra (ket) indices.

\begin{figure}[t]
    \centering
    \includegraphics[width=8cm]{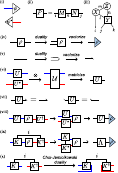}
    \captionsetup{justification=justified,singlelinecheck=false} 
    \caption{Graphical representation of some elementary mathematical objects pertaining to quantum dynamics. (i) Ket $\ket{\Psi}$ and corresponding bra $\bra{\Psi}$. The indices/wires of a matrix and its hermitian conjugate (adjoint) are colored blue and red, respectively. (ii) Matrix multiplication $P = MN$. (iii) A simple tensor network comprised of a contraction of three tensors $X, Y$ and $Z$. (iv) A density matrix $\rho$ in matrix and its vectorized form $\left|\rho \right\rangle \! \rangle$ by applying the state-functional duality, $|i\rangle \rightarrow \langle i|$. (In this paper, we represent \textit{normalized} density matrices as blue triangles.) (v) The Identity matrix $I$ (as a do-nothing operation on a wire) and its vectorized form $\left|I \right\rangle \! \rangle$.  (vi) Matrix form of a unitary map $\mathcal{U}(.) = U (.) U^\dagger$. (vii) The action of a unitary map $\mathcal{U}(.)$ on a the Identity, and (viii) on a density matrix $\rho$. (ix) A generic CP map with Kraus operators $\{K_i\}_i$. (x) The Choi-Jamiolkowski duality or isomorphism between CP maps and density matrices corresponds to applying the state-functional duality on the red and blue wires as shown.}
    \label{fig:graphical}
\end{figure}

Elementary tensor operations are also represented graphically. A matrix-vector multiplication, $w = Mv$,  is represented by connecting the bra wire of $M$ with the single (ket) wire of $v$. Analogously, the product of two matrices $P = MN$ is depicted by connecting together the ket wire of matrix $M$ and the bra wire of matrix $N$, Fig.~\ref{fig:graphical}(ii). The computational cost of multiplying a matrix $M_{ik}$ with a matrix $N_{kj}$ matrix is proportional to $|i||j||k|$.

Matrix multiplication can be generalized to the multiplication or contraction of any number of tensors. A contraction of a set of tensors can be specified by interconnecting their wires (consistent with their ket-bra orientations) according to a graph. Such a set of interconnected tensors is called a \textit{tensor network}. 
The computational cost of contracting a tensor network is proportional to the product of the dimensions of all the indices in the tensor network, generalizing the rule for estimating the computational cost of matrix multiplication. For example, the contraction cost of contracting the simple tensor network depicted in Fig.~\ref{fig:graphical} is proportional to $|i||j||k||l||m||n|$.

Tensor networks Refs.~\cite{Orus2019, Cirac2021, Bridgeman2017} are widely used in quantum many-body physics as efficient representations of complex quantum many-body wavefunctions. Examples of popular tensor networks include (1) Matrix Product States (MPSs), which consist of contracting tensors in a linear geometry, (2)  Tree Tensor Networks (TTNs) and the Multi-scale Entanglement Renormalization Ansatz (MERA), which are hierarchical tensor networks based on hyperbolic geometry, and (3) Projected Entangled Pair States (PEPSs), which generalize MPSs in higher-dimensions.

The graphical representation of the basic mathematical objects encountered when describing quantum processes is shown in Fig.~\ref{fig:graphical}. Here, we have used red and blue colors to distinguish between indices of elementary matrices, $M$, and their hermitian conjugate (adjoint), $M^\dagger$. The matrix forms of superoperators are not regarded as elementary matrices, and their indices -- which are tensor products of red and blue indices, and label a basis in the space $B(\mathcal{H})$ of bounded operators acting on a Hilbert space $\mathcal{H}$ -- are colored black.

A matrix $M_{ij}$ can be \textit{vectorized} into a vector $m_a$ by applying duality and grouping together the bra and ket indices, $i$ and $j$, into a single index $a$.\textsuperscript{\footnote{This group is realized by a canonical injective map  $(i,j) \rightarrow a$ such as the ones underlying the \texttt{reshape()} functions of Python and MATLAB.}} Figs.~\ref{fig:graphical}(iv)-(v) depict the vectorization of a density matrix $\rho$ and the Identity matrix, which is depicted simply as a wire.

A superoperator -- a map between matrices -- can be analogously \textit{matricized} by combining indices, as shown in Fig.~\ref{fig:graphical}(vi). Here, shown is the matricization of a unitary map $\mathcal{U} $ into a matrix -- which by slight abuse of notation, we call $\mathrm{U} = U \otimes U^*$ -- by grouping the left and right pairs of red and blue indices as black indices. The matricized map acts on vectorized Identity and density matrices, as shown in Fig.~\ref{fig:graphical}(vii)-(viii). Fig.~\ref{fig:graphical}(ix) depicts the matricization of a generic CP map with Kraus operators ${K_i}_i$. In the quantum dynamics literature, the matrix form of a CPTP map (for instance, the unitary map $\mathrm{U}$) is called the \textit{Liouville representation} of the map \cite{Wolf_undated-dn,wood_tensor_2015,OperationalQDynamics}.

\section{Translationally Invariant MPS Process} \label{ap:linear}
Here we review the generic correlation structure of a process with a restricted memory size. This will be in order to contrast the correlation structures exhibited by a general MPO process tensor discussed in past literature~\cite{processtensor,processtensor2,Gu2018,Gu2018A,Guo2020,white_non-markovian_2022}, with the process tree which we study in this work.

 In many practical situations, an unknown environment may affect a dynamics, but not carry over any significant memory (quantum information) to future interventions on the local system $\s$. Such processes are particularly amendable to numerical simulation, such as through an matrix product operator (MPO) ansatz for the corresponding process tensor. To showcase correlation behavior in this setting, consider a generic process which admits a (temporally) translation invariant MPO representation. That is, we define the MPO process $\Upsilon^{(\mathrm{MPO})}$, consisting of dynamics of a repeated CPTP map $\Lambda$ on a finite system, 
\begin{align}
    \Upsilon^{(\mathrm{MPO})}_{k}[\Lambda] &= \tr_E[\Lambda \star \dots \star \Lambda \star \rho] \\
    &= \includegraphics[scale=0.85, valign=c]{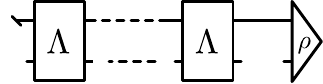}, \nn
\end{align}
where $\Lambda$ appears $n $ times.
It is well known that the repeated application of a generic CPTP map leads to exponentially decaying two-point correlations. In App.~\ref{ap:proof} we prove the following result to this to \emph{memory}, showing that non-Markovianity generically also decays exponentially, given a constant-sized environment.\textsuperscript{\footnote{This result could likely be extended to the situation of different CPTP maps, equivalent to noisy, time-dependent evolution; such as e.g. `ergodic sequences' of CPTP maps from Ref.~\cite{Movassagh2021}.}}
\begin{restatable}{prop}{MPS} \label{thm:MPS}
      In an MPO process $\Upsilon_{\mathrm{MPO}}[\Lambda]$ with a constant environment dimension $d_E$ and generic $\Lambda$, then in the asymptotic limit $n \delta t = t_b - t_a \to \infty$,
      \begin{equation}
          \eta(\Upsilon^{(\mathrm{MPO})}_{k}[\Lambda];t_a,t_b) \sim \ex^{- |t_b - t_a|/{\xi}},
      \end{equation}
      where $\xi >0$.
\end{restatable}
We now offer a sketch of a proof of this result. For a generic (irreducible) CPTP map $\Lambda$, for large $\Delta t$, two-point connected correlations decay exponentially. As this is true for any choice of observables, this directly leads to the states $\ups_{b,a}$ and $\ups_a \otimes \ups_b$ being exponentially close to each other, according to operational meaning of the one-norm distance. We can then apply the Fannes-Aubenaert inequality~\cite{Wilde_2011}, which roughly states that when two density matrices are close, so are their entropies. This difference of entropies $S(\ups_a \otimes \ups_b)-S(\ups_{b,a})$ is exactly the quantum mutual information measure of non-Markovianity, $ \eta(\ups; t_a, t_b)$ from Eq.\eqref{eq:NM2}.

The two key properties leading to this exponential decay property here is that (i) the process $\Upsilon_{\mathrm{MPO}}[\Lambda]$ has a bounded size memory, and (ii) it has a linear (MPO) geometry. This is analogous to the state case, where an area law -- a bounded bond dimension -- is generically equivalent to exponentially decaying correlations~\cite{Eisert2010area,brandao2015exponential}. Of course, a perfectly coherent system-environment (unitary) interaction with a constant sized environment would lead to recurrences and hence non-monotonically decaying temporal correlations. Such interactions are rather atypical in a practical sense, given the difficulty in isolating a quantum system from outside noise. For noisy interactions (represented by CPTP maps), Prop.~\ref{thm:MPS} indicates that a finite sized environment generically leads to exponentially decaying temporal correlations. However, within this translationally invariant (linear) MPS geometry, it is not possible to achieve all possible correlation structures. There are relevant physical examples that exhibit a strong memory for all times, which can not be efficiently modeled with an MPO processes, for example the spin-boson model which we examine in Sec.~\ref{sec:grug}. This motivates our introduction of the process tree, which utilizes tensor-tree-like geometry to produce multitime processes with structurally slowly decaying correlations and memory.

\section{Formal Properties of General Process Trees} \label{ap:y-type}
\begin{figure}[t]
    \centering
    \includegraphics[width=\columnwidth]{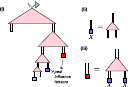}
    \captionsetup{justification=justified,singlelinecheck=false} 
    \caption{Without the scale consistency condition Eq.~\eqref{eq:isometry}, contracting the process tree in computing correlation functions is not as efficient, as one generally needs to compute non-trivial influence tensors. old: (iii) The tensor network contraction that equates to the expectation value of the intervention $A$ in a $\mc{Y}$-type process tree. }
    \label{fig:influence_tensors}
\end{figure}
In this section, we discuss additional some properties of one-to-two time superprocesses, extending the exposition of Sec.~\ref{sec:treeconstruct}. For a $\mc{Y}-$type process tree, i.e. that without the scale consistency condition Eq.~\eqref{eq:isometry}, one needs to keep track of the effect of inserting identity superoperators $\cup$ across times without observables/instruments, and across fine time scales.

General process trees instead are in terms of a general one-to-two superprocess as in Eq.~\eqref{eq:y_brick}. As discussed in Sec.~\ref{sec:coarse}, such a transformation does not need preserve the identity instrument when coarse graining up the process tree, $\mc{Y}^{\mathrm{T}}|{\cup }\rrangle = |{X}\rrangle $, where $\mathrm{X} \neq \cup$ in general. Therefore, even a single-time correlation function could in principle depend on all $O(2^N)$ tensors $\{ \mc{Y}^{(j)}_s\}$ in the tree. This fact is represented by the scale causal cone shown in Fig.~\ref{fig:causalcone1} (ii), in comparison to $\mc{W}-$class bricks as shown in Fig.~\ref{fig:causalcone1} (i). When computing correlation functions, one needs to keep track of a range of `influence tensors', as shown in Fig.~\ref{fig:influence_tensors}. This could still, in-principle, be done relatively efficiently for a homogeneous process tree. However, one needs to carefully choose the height of the process tree $N$, as such a choice would have a significant impact on observables, in contrast to the $\mc{W}-$brick case.

\subsection{Proof that $\mc{Y}$ is a superprocess} \label{ap:superprocess}
We here prove explicitly that $\mathcal{Y}$ as defined in Eq.~\eqref{eq:y_brick} is a valid superprocess. To prove this, we need to show that when acting on any step of an arbitrary process $\ups_{k}$, the output $\ups_{k+1}^\prime = \mathcal{Y} \ups_{k}$ satisfies (i) positivity and normalization Eq.~\eqref{eq:normalization} and (ii) the affine causality conditions Eq.~\eqref{eq:causality}.

Recall the general expression for a process Eq.~\eqref{eq:pt_link},
\begin{equation}
    \ups_{k}=\includegraphics[scale=0.85, valign=c]{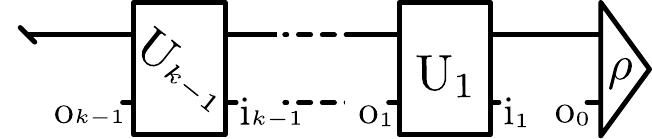}. 
\end{equation}
and for the general $\mc{Y}$ brick,
\begin{equation}
        \mathcal{Y}^{\mathrm{i}_c \mathrm{o}_c}_{\mathrm{i}_{f} \mathrm{o}_{f} \mathrm{i}_{\bar{f}} \mathrm{o}_{\bar{f}} }=\includegraphics[scale=1.5, valign=c]{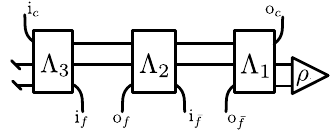}.  \label{eq:w_brick1}
\end{equation}
Then we can write out the full expression of a single fine-graining map acting on a time slot of an arbitrary process tensor $\ups_{k}$
\begin{align}
       (\ups_{k+1}^\prime)^{ \mathrm{o}_k \mathrm{i}_{k} \dots \mathrm{i}_1 \mathrm{o}_0}_{\mathrm{i}_{f} \mathrm{o}_{f} \mathrm{i}_{\bar{f}} \mathrm{o}_{\bar{f}} } &=\mathcal{Y}^{\mathrm{i}_c \mathrm{o}_c}_{\mathrm{i}_{f} \mathrm{o}_{f} \mathrm{i}_{\bar{f}} \mathrm{o}_{\bar{f}} } ( \ups_{k} )^{ \mathrm{o}_k \mathrm{i}_{k} \dots \mathrm{i}_c \mathrm{o}_c \dots \mathrm{i}_1 \mathrm{o}_0} \nn\\
       &= \includegraphics[scale=0.6, valign=c]{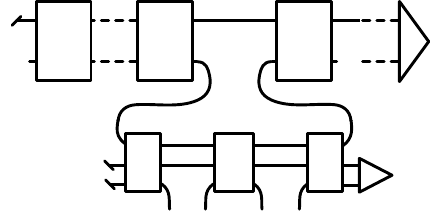} \label{eq:upsP}
\end{align}
where we have dropped the labels in the graphical representation in the second line for comprehensibility. Then the two conditions detailed above can be directly proven from this expression:
\begin{enumerate}[label=\text{(\roman*)}.]
    \item Eq.~\eqref{eq:upsP} is a network of a composition of CPTP maps, and so the  positivity and normalization conditions of Eq.~\eqref{eq:normalization} are satisfied.
    \item The causality conditions Eq.~\eqref{eq:causality} for $ (\ups_{k+1}^\prime)$ are inherited from $ (\ups_{k})$ for outputs at times $\mathrm{o}_{c+j}$ for $j \geq 2$. We can then graphically show that the conditions are also preserved on the $f, \bar{f}$ indices, 
    \begin{align}
        \tr_{\mathrm{o}_{c+1}}[ \ups_{c+1:0}^\prime] &= \includegraphics[scale=0.6, valign=c]{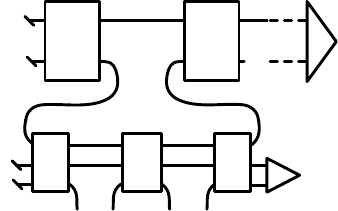} \nn \\
        &= \includegraphics[scale=0.6, valign=c]{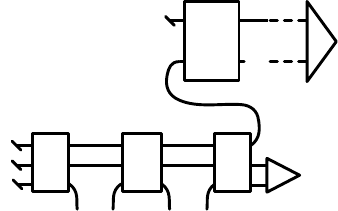}  \\
        &= \includegraphics[scale=0.6, valign=c]{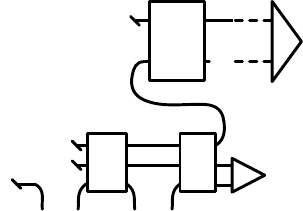} \nn \\
        &= | \id_{\mathrm{i}_{f}} \rangle \! \rangle \otimes \ups_{f:0}^\prime,\nn
    \end{align}
    where we have used that trace commutes through TP maps from the left. We can then apply the same techniques to further prove that 
    \begin{align}
        &\tr_{\mathrm{o}_{{f}}}[ \ups_{f:0}^\prime]= | \id_{\mathrm{i}_{\bar{f}}} \rangle \! \rangle \otimes \ups_{\bar{f}:0}^\prime, \text{ and} \\
        &\tr_{\mathrm{o}_{{\bar{f}}}}[ \ups_{\bar{f}:0}^\prime]= | \id_{\mathrm{i}_{\bar{f}}} \rangle \! \rangle \otimes \ups_{c-1:0}^\prime = | \id_{\mathrm{i}_{c-1}} \rangle \! \rangle \otimes \ups_{c-1:0}.\nn
    \end{align}
\end{enumerate}
Therefore $\mc{W}$ is a valid superprocess, and so the process tree in Fig.~\ref{fig:generaltree} is a process tensor. 

\subsection{Gauge Freedom in $\mc{W}-$type Process Trees}
Here we describe that we can change the scale consistency condition to instead be with respect to any other unitary map, by only modifying a $\mc{W}-$class process tree through local unitary operations on the finest ($s=N$) and coarsest scales ($s=0$). 

We start with a $\mc{W}-$class superprocess, satisfying Eq.~\eqref{eq:isometry}. Consider an arbitrary unitary matrix $V$, where as usual its Liouville superoperator representation is $\mathrm{V}= V \otimes V^*$ (see Sec.~\ref{sec:background}). We modify the $\mc{W}-$brick as follows 
\begin{equation}
    \mc{W}_{\mathrm{V}} := \includegraphics[scale=2.5, valign=c]{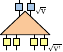},
\end{equation}
where the yellow (bottom) boxes represent $\sqrt{\mathrm{V}^{{\dagger}}}$, while the blue (top) boxes represent $\sqrt{\mathrm{V}}$. Then $\mc{W}_{\mathrm{V}}$ satisfies a new scale consistency condition with respect to the unitary instrument $\mathrm{V}$, instead of the identity supermap $\cup$ as in Eq.~\eqref{eq:isometry},
\begin{equation}
    \mc{W}_{\mathrm{V}} |{ \mathrm{V}}\rrangle |{ \mathrm{V}}\rrangle = |{ \mathrm{V}}\rrangle.
\end{equation}
One can see this directly from the fact that $\sqrt{\mathrm{V}^{{\dagger}}} \mathrm{V} \sqrt{\mathrm{V}^{{\dagger}}} =\cup$ (the identity map), and then by using the property Eq.~\eqref{eq:isometry} for $\mc{W}$. Constructing an entire process tree from $\mc{W}_{\mathrm{V}}$ leads to canceling of the blue and yellow blocks ($\sqrt{\mathrm{V}}$ and $\sqrt{\mathrm{V}^{{\dagger}}}$ within the bulk, and is equivalent to the usual a $\mc{W}-$class process tree, except by yellow bricks at every site on the finest scale ($s=N$), and blue bricks at the coarsest.

\section{Proofs} \label{ap:proof}
\polyDecay*

\begin{proof}
Let $\ups = \tree_N $ for some $N$. Consider temporal correlation functions at a duration of $\Delta t = t_n-t_0=2^n-1$ steps, where $n \in \mathbb{Z}_{\geq 1}$. For simplicity we also here we take $t_0=0$ to be on the leftmost part of the Tree (earliest time), and so $t_n = 2^n-1$. Graphically, using the scale consistency condition \eqref{eq:isometry} these correlation functions reduce to 
    \begin{align}
    &\braket{A_{t_2}^\prime A_{t_0}}_{\ups}=\includegraphics[scale=1.2, valign=c]{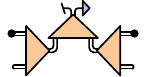},\nn \\
    & \braket{A_{t_3}^\prime A_{t_0}}_{\ups}=\includegraphics[scale=1.2, valign=c]{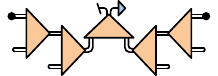}, \\
    & \braket{A_{t_n}^\prime A_{t_0}}_{\ups} \!=\!\includegraphics[scale=1.2, valign=c]{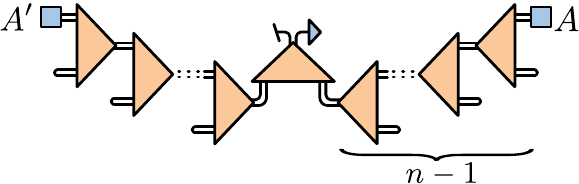} .\nn
\end{align}
In particular, we see a repeated structure with increasing $n$. This is the key expression of the proof. We can see that this is a function of a transfer matrix $\mathcal{W}_R$,
\begin{equation}
    \braket{A_{t_n}^\prime A_{t_0} }_{\ups} = \llangle{A^\prime} | \mathcal{W}_L^{n-1} \mc{W}^{(\mathrm{b})}  (\mathcal{W}_R^{{T}})^{n-1} |A\rrangle
\end{equation}
where the superscript $T$ denotes transpose, and 
\begin{equation}
    \mathcal{W}_R:= \includegraphics[scale=1.2, valign=c]{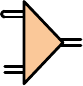}, \, \mathcal{W}_L:= \includegraphics[scale=1.2, valign=c]{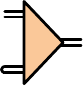}, \, \mathrm{W}^{(\mathrm{b})} := \includegraphics[scale=1.2, valign=c]{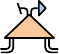}. \label{eq:Tdef}
\end{equation}
The first two matrices are what we call right/left descending maps, and are key elementary operations in computing correlation functions  numerically, as we discuss in App.~\ref{ap:numerics}. We assume that the process tree is deep enough ($N \gg 0)$ such that the contribution from the middle tensor, $\mc{W}^{\mathrm{(b)}}$, is approximately constant with increasing $n$ (and hence $\Delta t$). The task is then to show that this transfer matrix $\mathcal{W}_R$ is dominated by a subleading eigenvector as $n \to \infty$, and that the corresponding eigenvalue satisfies $|\lambda | < 1$. Then we can show that the connected correlation function $|\braket{A_{t_0} A_t}_{\ups}- \braket{A_{t_0}}_{\ups} \braket{A_t}_{\ups}|  $ decays exponentially with $n$, and therefore polynomially so with $\ell$. To prove this, we need two ingredients: complete positivity (CP) and the existence of a (unique) eigenvector with eigenvalue $1$. Then the Russo-Dye Theorem ensures that all eigenvalues lie within the complex unit disk $| \lambda | \leq 1$~\cite{Wolf_undated-dn}, and that there exists a single dominant eigenvector with the maximal $\lambda  = 1$ (see the proof for Thm. 2.3.7 in Ref.~\cite{Bhatia2007}).

 Choosing the intervention separation $ \Delta t =2^n-1$ ensures that the causal cone of each intervention consists only of the left or right descending maps. Therefore, the correlator Eq.~\eqref{eq:corr_function} is dictated by the $(n-1)^{\mathrm{ th}}$ power of $\mc{W}_L$ and $\mc{W}_R^{\mathrm{T}}$, in each causal cone.

We will first show that $\mathcal{W}_R^{\mathrm{T}}$ is CP. We will do this by showing that the Liouville superoperator $\mathcal{W}_R^{\mathrm{T}}$ can be written in terms of the Kraus operators $K_i$ as
\begin{equation}
    \mathcal{W}_R^{\mathrm{T}} = \sum_i K_i \ot K^*_i. \label{eq:Kraus}  
\end{equation}
Here, $K_i$ acts on the ket parts of the input, and $K_i^*$ on the bra, with the index $i$ being generated from all lines connecting the two `halves'. Being able to write the Liouville superoperator representation of a map in this way is equivalent to the map admitting a Kraus operator sum representation (with equal left and right operators), and this is possible only for CP maps~\cite{Wolf_undated-dn,wood_tensor_2015,OperationalQDynamics}.

To prove Eq.~\eqref{eq:Kraus} we will use the definition of $\mathcal{W}_R^{\mathrm{T}}$, Eq.~\eqref{eq:Tdef}, in terms of the $\mathcal{W}$  superprocess definition Eq.~\eqref{eq:w_brick}. Explicitly, 
\begin{equation}
    \begin{split}
        \mathcal{W}_R^{\mathrm{T}}&= \includegraphics[scale=1.2, valign=c]{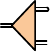}  \\
    &=\includegraphics[scale=1.5, valign=c]{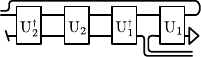}  \\
    &=\includegraphics[scale=1.5, valign=c]{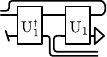}  \\
    &=\includegraphics[scale=1.5, valign=c]{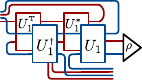}  \label{eq:Tkraus}\\
    &=\sum_{i_1,i_2} \bra{i_2}\rho \ket{i_2} \left( \includegraphics[scale=1.5, valign=c]{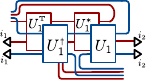} \right)  \\
    &= \sum_{\vec{i}=(i_1,i_2)} K_{\vec{i}} \ot K^*_{\vec{i}}, 
    \end{split}
\end{equation}
where in the fourth line we have `unfolded' the tensor network diagram, such that wires represent a bra/ket contraction, and boxes are unitary matrices with italic labels. This is in contrast to the two preceding lines, where boxes represent the superoperator representation of a unitary map with upright font labels; see Eq.~\eqref{eq:PT} and the surrounding exposition. For clarity, we have coloured all the ket objects in red, and the bra objects in blue. The trace at the end then corresponds to connecting the ket and bra copies through the orthonormal basis of projections $\bra{i_2}$, and the initial state $\rho =  \sum_{i_2} \bra{i_2}\rho \ket{i_2}  \ket{i_2} \bra{i_2}$ provides another connection between the copies, with $\ket{i_2}$ being the (diagonal) eigenbasis of $\rho$. With $\vec{i}$ in Eq.~\eqref{eq:Tkraus} being the combined index of $(i_1,i_2)$,  
\begin{equation}
    K_{\vec{i}}=\sqrt{\bra{i_2}\rho \ket{i_2}} \left( \includegraphics[scale=1.4, valign=c]{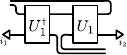}\right),
\end{equation}
and so $\mathcal{W}_R^{\mathrm{T}}$ can be written in the form of Eq.~\eqref{eq:Kraus}. $\mathcal{W}_R^{\mathrm{T}}$ is therefore CP~\cite{OperationalQDynamics}. Note that this does not depend on the particular choice of $\mc{W}$ in terms of $\mathrm{U}_1$ and $\mathrm{U}_2$, compared to other maps, and a similar proof to \eqref{eq:Tkraus} would hold for other choices of $\mc{W}$ satisfying the scale consistency condition Eq.~\eqref{eq:isometry}.

Now, due to the scale consistency condition \eqref{eq:isometry}, we know there exists an eigenvector with eigenvalue one. This is the identity superoperator $| \cup \rrangle$,
\begin{equation}
   \mathcal{W}_R^{\mathrm{T}} | \cup \rrangle = \includegraphics[scale=1.2, valign=c]{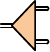}= \includegraphics[scale=1.2, valign=c]{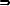} = | \cup \rrangle. 
\end{equation}
We will assume for now that this {is} the unique eigenvector (i.e. that $\mathcal{W}_R^{\mathrm{T}}$ is irreducible~\cite{Wolf_undated-dn}), and later argue that this is the generic case. Now, an equivalent argument to that above applies to the case of $\mathcal{W}_L$ via left-contraction (left matrix multiplication), such that it is CP with a unique left eigenvector.

As $| \cup \rrangle$ is a right eigenvector with maximum eigenvalue $\lambda_{\cup}=1$, we can write 
{\begin{equation}
    \lim_{n \to \infty} (\mathcal{W}_R^{\mathrm{T}})^{n-1} \approx | \cup \rrangle\llangle{\lambda_R}|,
\end{equation}}
where we leave the left eigenvector $|{\lambda}_R\rrangle$ unspecified, but know that $\llangle{\lambda_R | \cup}\rrangle = 1$~\cite{Wolf_undated-dn}. An equivalent statement applies to $\mathcal{W}_L$ acting through left multiplication.
In this case, for large $n$, to first order two-point correlations for the process tree factorize,
\begin{align}
    \lim_{n \to \infty} \braket{A A_t}_{\ups} &= \lim_{n \to \infty} \llangle A | (\mathcal{W}_L)^{n-1} \mc{W}^{(\mathrm{b})}  (\mathcal{W}_R^{\mathrm{T}})^{n-1} |A\rrangle \nn \\
    &\approx \llangle{A | \lambda_L}\rrangle \llangle \cup |  \mc{W}^{(\mathrm{b})} | \cup \rrangle \llangle{\lambda_R | A}\rrangle \\
    &=\llangle{A | \lambda_L}\rrangle\llangle \cup |  \mc{W}^{(\mathrm{b})} | \cup \rrangle \llangle \cup |  \mc{W}^{(\mathrm{b})} | \cup \rrangle \llangle{\lambda_R | A}\rrangle \nn\\
    &=\lim_{n \to \infty} \llangle A | \mathcal{W}_L^{n-1}  \mc{W}^{(\mathrm{b})} | \cup \rrangle \llangle \cup |  \mc{W}^{(\mathrm{b})} (\mathcal{W}_R^{\mathrm{T}})^{n-1} |{ A}\rrangle \nn \\
    &= \braket{A_{t}}_{\ups} \! \braket{A_{t_0}}_{\ups}. \nn
\end{align}
Here, in the third line we have used that $\llangle \cup |  \mc{W}^{(\mathrm{b})} | \cup \rrangle = 1$, which can be proven from the scale consistency condition \eqref{eq:isometry} and from the normalization of the initial environment state,
\begin{equation}
    \llangle \cup |  \mc{W}^{(\mathrm{b})} | \cup \rrangle =  \includegraphics[scale=1.2, valign=c]{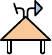} = \includegraphics[scale=1.2, valign=c]{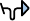} = \tr[\rho_0] = 1.
\end{equation}
Therefore, the connected correlation function will decay according to the subleading eigenvalues $\lambda_{2} = \lambda_{2}^L \lambda_{2}^R$,
\begin{equation}
    |\! \braket{A_{t_0} A_t}_{\ups}\! - \braket{A_{t}}_{\ups} \! \braket{A_{t_0}}_{\ups}\!| \underset{\ell \gg 0}{\sim} |\lambda_{2}|^{2n}  =   e^{- n/\xi }, 
\end{equation}
where the final equality is due to $|\lambda_{\mathrm{sub}}|<1$ and where $\xi > 0 $ can be interpreted as the (temporal) correlation length. Using $\Delta t = 2^{n}-1$, we arrive at the result.

Randomly sampled matrices generically have a non-degenerate spectrum. The question remaining is whether the internal structure of the transfer matrix $\mathrm{W}_{R/L}$ means that it will not satisfy this same generic property. For choices of randomly chosen $U_1$ and $U_2$ within the definition of $\mathrm{W}$, we numerically compute the spectrum of $\mathrm{W}_{R/L}$ and find that it tends to always be non-degenerate. 
\end{proof}

\MPS*

\begin{proof}
    Assume that the CPTP dynamics $\Lambda$ is generic, in that it is irreducible and so admits a unique (left) eigenvector with eigenvalue $|\lambda_1 | = 1$, from its TP property~\cite{Wolf_undated-dn}. Then, for large times $ t_b - t_a = n \delta t$,  
    \begin{equation}
        \Lambda^n \approx |{\lambda}\rrangle \llangle {\id}| + \lambda_{2}^n P_2 +\mc{O}(\lambda_3),
    \end{equation}
    where $\lambda_2$ is the sub-leading eigenvalue, and $P_2$ the corresponding projector onto this Jordan block. From this, one can show using the Frobenius Perron Theorem that two-point correlators will have exponentially decay with (a large) number of `time steps' $n$~\cite{Wolf_undated-dn},
    \begin{equation}
        |\braket{A(t_a) B(t_b)}- \braket{A(t_a)}  \braket{B(t_b)} | \sim \ex^{-n/\xi}. \label{eq:2pt-corr-MPS}
    \end{equation}
    Now consider any POVM $\mc{M}$, which is defined by the set of positive Hermitian operators $\{ M^{(r)}\}$, called POVM elements, such that $\sum_r M_r =\id$. Any arbitrary POVM element on a bipartite space can be written in a basis of separable (local) operators,
    \begin{equation}
        M^{(r)} = \sum_i {A}^{(r)}_i \otimes {B}^{(r)}_i,
    \end{equation}
    where ${A}^{(r)}_i$ and ${B}^{(r)}_i$ are operators solely on the space $\mc{H}_{t_a}$ and $\mc{H}_{t_b}$ respectively. This comes from the fact that any operator can be expanded in a local operator basis (such as the Pauli basis). Such decomposition always exists, which can be proven from operator Schmidt decomposition. Recall the operational definition of the $1-$norm distance,
    \begin{equation}
        \| \rho - \sigma\|_1 := \underset{\mc{M}}{\max } \sum_r | \tr[M^{(r)} (\rho - \sigma)]|,
    \end{equation}
    where the maximum is over all POVMs. Writing this out for $\ups_{b,a}$ as in Eq.~\eqref{eq:NM2}, and expanding the POVMs in a local basis on $\mc{H}_{t_a} \otimes \mc{H}_{t_b}$
    \begin{align}
        \| \ups_{b,a} - &\ups_a \! \otimes \! \ups_b \|_1 = \underset{\mc{M}}{\max } \sum_r | \tr[M^{(r)} (\ups_{b,a} - \ups_a\! \otimes \!\ups_b)]| \nn\\
        &= \underset{\mc{M}}{\max } \sum_r | \sum_i \tr[{A}^{(r)}_i \otimes {B}^{(r)}_i (\ups_{b,a} - \ups_a \otimes \ups_b)]|\nn \\
        &\leq \underset{\mc{M}}{\max } \sum_r \sum_i  | \tr[{A}^{(r)}_i \otimes {B}^{(r)}_i (\ups_{b,a} - \ups_a \otimes \ups_b)]| \nn \\
        &=\underset{\mc{M}}{\max } \sum_r \sum_i | \braket{{A}^{(r)}_i {B}^{(r)}_i} - \braket{{A}^{(r)}_i} \braket{{B}^{(r)}_i} |, \nn 
    \end{align}
    where in the penultimate line we have used the triangle inequality $| \sum_i x_i| \leq \sum_{i} |x_i | $. Now, the inside of the absolute value is exactly a connected two-point function Eq.~\eqref{eq:2pt-corr-MPS}, which decays exponentially with $n$ for any operators and so any $r$, $i$. The summation gives a multiplicative factor that depends only on the dimension of $\mc{H}_{t_a} \otimes \mc{H}_{t_b}$, and so 
    \begin{equation}
         \| \ups_{b,a} - \ups_a \! \otimes \! \ups_b \|_1 \sim \ex^{-n/\xi}.
    \end{equation}
    Finally, we can apply the Fannes-Aubenaert inequality to get~\cite{Wilde_2011}
    \begin{align}
        \eta(\ups; t_a, t_b)&=S(\ups_{a})+S(\ups_{b})-S(\ups_{b,a})\\
        &=S(\ups_{a} \otimes \ups_{b})-S(\ups_{b,a}) \\
        &\leq \ex^{-n/\xi} (n/\xi + \kappa \log d_{ab} ),
    \end{align}
    where in the second line we have also used the additivity of von Neumann entropy, $d_{ab}$ is the dimension of the density matrix $\ups_{b,a}$, and where $\kappa$ is a positive constant. With long times $n \to \infty$ the exponential factor here dominates, and so non-Markovianity as measured by the quantum mutual information $\eta(\ups; t_a, t_b)$ must decay at least exponentially with large $n$. 
\end{proof}

\section{Details of Process Tree Numerics} \label{ap:numerics}
This section serves as extra details on the numerical simulations computed for the results in Fig.~\ref{fig:numerics}, as well as an explanation for how one can estimate the computing cost for a correlation function computed on a process tree.

Arbitrary correlations in a $\mc{W}-$class process tree can be computed numerically through three elementary operations: left and right descending moves as in Eq.~\eqref{eq:Tdefa}, and a fusion move which we will introduce below. 

We first return to the computation of a one-time expectation value of an intervention $A$ acting on the $t^{\mbox{\tiny th}}$ slot of a process tree with $2^N$ intervention slots, as in Sec.~\ref{sec:one-time}. For simplicity, we assume that each wire of the tensor network has the same dimension $d$. We can break down the total contraction into a sequence of two basic contractions corresponding to coarse-graining an instrument acting either on the left or right slot of a tensor:
\begin{align}
    &\mathcal{W}_R |A\rrangle=\includegraphics[scale=3, valign=c]{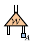}=\includegraphics[scale=3, valign=c]{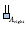}, \nn\\
    &\mathcal{W}_L |A\rrangle=\includegraphics[scale=3, valign=c]{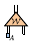}=\includegraphics[scale=3, valign=c]{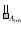}. \label{eq:Aright}
\end{align}
We refer to these basic contractions as the \textit{left} and \textit{right} moves, respectively. The contraction to evaluate the one-time expectation value can now be viewed as a sequence of left or right moves that coarse-grains the intervention $A$ all the way to the longest scale, where it is contracted with the single-time measurement process $\tree_0$ (Eq.~\eqref{eq:prepare}). 

According to the general estimate of the cost of contracting tensor networks reviewed in App.~\ref{ap:tn}, the computational cost of a left or right move is proportional to $d^6$. Since the scale causal cone consists of precisely one tensor at every time scale (Fig.\ref{fig:causalcone1}), the total number of left or right moves required to coarse-grain intervention $A$ to the longest time scale is $N$ (equal to the height of the process tree). Therefore, the total computational cost of computing the expectation value of an instrument that acts on a single intervention slot is proportional to $N d^6$, as reported in Sec.~\ref{sec:one-time}. 

\begin{figure}[t]
    \centering
    \includegraphics[width=8cm]{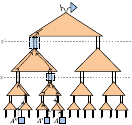} \captionsetup{justification=justified,singlelinecheck=false}
    \caption{A diagrammatic representation of the contractions involved in computing three-point correlations in a process tree. Analogous to Fig.~\ref{fig:causalcone2}, by coarse-graining a three-time correlator of interventions $A, A'$ and $A''$ simplifies to a two-time correlator at some scale $s$ and then a one-time expectation value at a longer scale $s' > s$. }
    \label{fig:causalcone3}
\end{figure}

The decomposition of the total contraction into a sequence of left and right moves, as described above, also eases the software implementation of the computation of one-time expectation values, particularly if the process tree is homogeneous (i.e. $\mc{W}^{(j)}_s\equiv \mc{W}$). In practice, one only has to implement the basic contractions, Eq.~\eqref{eq:Aright}, and then call them iteratively to coarse-grain the intervention to the longest time scale, where it is easily contracted with single-time measurement process (Eq.~\eqref{eq:prepare}). This can be done through a binary expansion of the site number $t$. This is easiest to explain via an example. Recall the single time expectation value Eq.~\eqref{eq:1-time} from the main text,
\begin{align}
    \braket{A}_{\tree}&=\includegraphics[scale=1.5, valign=c]{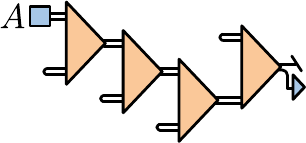} \label{eq:1-time-again}\\
    &= \llangle A |    \mc{W}^{(4)}_{4,L} \mc{W}^{(2)}_{3,L} \mc{W}^{(1)}_{2,L} \mc{W}^{(1)}_{1,R} |{\rho}\rrangle |{\id}\rrangle.
\end{align}
This comes from computing the expectation value of the instrument $A$ on the site $t=8$, for the height $N=4$ tree as in Fig.~\ref{fig:causalcone1} (i). Then, in binary $t=\text{`}1000\text{'}$ and from this we can immediately determine the required tensors in the contraction in Eq.~\eqref{eq:1-time-again}: reading from left to right, the first digit $\text{`}1\text{'}$ corresponds to contracting with $\mc{W}_R$ first, then the next three $\text{`}0\text{'}$ digits correspond to three left moves $\mc{W}_L$.

The computational cost for two-time and higher correlation functions can be estimated analogously by breaking down the computation into a sequence of elementary contractions. Each instrument independently ascends through its causal cone, via the left or right move, Eq.~\eqref{eq:Aright}, and the two instruments fuse together into an effective one-time intervention at some scale $s$, which then ascends the causal cone to the top of the tree. The \textit{fusion move} corresponds to the contraction of a  tensor with both slots occupied by instruments:
\begin{equation}
\mc{W} |{A^\prime}\rrangle |A\rrangle = \includegraphics[scale=3, valign=c]{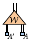} = \includegraphics[scale=3, valign=c]{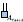} = |{A_{\mathrm{fused}}} \rrangle.\label{eq:Afused}
\end{equation}
One can implement this numerically through the binary expansion of the time $t$ and $t^\prime$. Each index represents a left or right move, and where the indices first differ (reading from left to right) is where a fusion move is applied. We implement these operations above to arrive at the efficient numerical results of Fig.~\ref{fig:numerics}, using the Python package \texttt{ncon} for elementary tensor network semantics~\cite{pfeifer2015ncon}.

Finally, higher-time correlators then generalize the above in a straightforward manner, and can also be computed as a sequence of the left, right, and fusion moves. Fig.~\ref{fig:causalcone3} shows the tensor network contraction for a three-time correlator. From this figure, it is intuitive to see that the polynomial decay of correlations results of Fig.~\ref{fig:numerics} and Thm.\ref{thm:polyDecay}, in terms of $\Delta t$, will generalize to $k-$time correlations $\braket{A_{t_k }  \dots A_{t_2 } A_{t_1 }  A_{t_0}}_{\tree}$ for $k \geq 3$ to produce polynomial decay in each of the variables $\Delta t_i = t_{i+1} - t_i $.

\end{document}